\documentclass[aps,pra,reprint,superscriptaddress,amsmath,preprintnumbers,nofootinbib,floatfix]{revtex4-2}

\usepackage[utf8]{inputenc}
\usepackage[shortcuts]{extdash}
\usepackage[english]{babel}

\usepackage{physics}
\usepackage{amssymb,amsmath,dsfont,amsthm}
\usepackage{graphicx}
\usepackage{natbib}

\usepackage{enumitem}
\usepackage{svg}
\usepackage{bm}
\usepackage{xcolor}

\usepackage{csvsimple}
\usepackage{booktabs}
\usepackage[caption=false]{subfig}
\usepackage{tabularx}
\usepackage{ulem}

\usepackage{hyperref}
\usepackage[capitalise]{cleveref}
\usepackage{booktabs}
\usepackage{braket}
\usepackage{comment}
\usepackage{tabularx}

\DeclareMathOperator\artanh{artanh}

\newtheorem{proposition}{Proposition}

\newcounter{protocol}

\newcommand{\id}{\mathds{1}}

\begin{document}

\title{Towards Device-Independent Quantum Key Distribution with Photonic Devices}

\author{Corentin Lanore}
\affiliation{Universit\'e Paris-Saclay, CEA, CNRS, Institut de physique th\'eorique, 91191, Gif-sur-Yvette, France}

\author{Xavier Valcarce}
\email{xavier.valcarce@ipht.fr}
\affiliation{Universit\'e Paris-Saclay, CEA, CNRS, Institut de physique th\'eorique, 91191, Gif-sur-Yvette, France}

\author{Jean Etesse}
\affiliation{Université Côte d’Azur, CNRS, Institut de Physique de Nice (INPHYNI), Nice, France}

\author{Anthony Martin}
\affiliation{Université Côte d’Azur, CNRS, Institut de Physique de Nice (INPHYNI), Nice, France}

\author{Jean-Daniel Bancal}
\affiliation{Universit\'e Paris-Saclay, CEA, CNRS, Institut de physique th\'eorique, 91191, Gif-sur-Yvette, France}

\begin{abstract}
Quantum Key Distribution (QKD) protocols enable two distant parties to communicate with information-theoretically proven secrecy. However, these protocols are generally vulnerable to potential mismatches between the physical modeling and the implementation of their quantum operations, thereby opening opportunities for side channel attacks.
Device-Independent (DI) QKD addresses this problem by reducing the degree of device modeling to a black-box setting.
The stronger security obtained in this way comes at the cost of a reduced noise tolerance, rendering experimental demonstrations more challenging: so far, only one experiment based on trapped ions was able to successfully generate a secret key.
Photonic platforms have however long been preferred for QKD thanks to their suitability to optical fiber transmission, high repetition rates, readily available hardware, and potential for circuit integration.
In this work, we assess the feasibility of DIQKD on a photonic circuit recently identified by machine learning techniques.
For this, we introduce an efficient converging hierarchy of semi-definite programs (SDP) to bound the conditional von Neumann entropy and develop a finite-statistics analysis that takes into account full outcome statistics. Our analysis shows that the proposed optical circuit is sufficiently resistant to noise to make an experimental realization realistic.
\end{abstract}


\maketitle

\section{Introduction} \label{sec:intro}
Quantum theory predicts that measurements on two quantum systems prepared in a maximally entangled state can yield perfectly correlated outcomes that are fundamentally unpredictable to any third party. This is one of the founding ideas of Quantum Key Distribution (QKD)~\cite{Ekert91}, which promises secure communication even in the presence of an eavesdropper with unlimited computational capability.

However, the security of QKD schemes typically relies on some level of characterization of the quantum hardware, whether for its quantum sources or measurement devices~\cite{Scarani2009,renner_quantum_2023}. Deviation of the quantum hardware from its expected behavior can then be exploited by an adversary to attack the protocol, thereby compromising its security~\cite{makarov-nature,makarov-nature2,review-pan}. Device-Independent Quantum Key Distribution (DIQKD) protocols address this issue by relaxing the requirements on quantum devices and treating them simply as black boxes~\cite{Acn2007,Vazirani2014,ArnonFriedman2018}. By removing the need to trust the internal working of the quantum devices employed, these protocols reduce underlying assumptions and enhance security. The proper functioning of the quantum devices is then inferred during the execution of the protocol by monitoring a Bell score, rather than being trusted a priori~\cite{reviewBell}. This ensures that the parties' measurement outcomes remain secret regardless of the device's internal operations~\cite{Pironio2010,DIQKDexp2}.

The stronger security guarantees of DIQKD protocols come at the cost of a challenging experimental implementation. Achieving Bell inequality violations high enough to extract a secret key while simultaneously closing the detection loophole \cite{reviewBell} is a tall order. It is therefore not surprising that the first DIQKD experiments have only appeared recently~\cite{DIQKDexp2,DIQKDexp1,DIQKDexp3}. One such experiment employed a complex setup combining trapped ions and entanglement swapping~\cite{DIQKDexp2}. This system successfully realized a DIQKD protocol, generating a secure key between two parties distant by $2$ meters. The neutral atom platform used in~\cite{DIQKDexp1} exhibited comparable complexity; however, in this case, the larger distance prevented extracting a key, once again demonstrating the difficulty of implementing a DIQKD protocol. A third experiment based on a photonic platform~\cite{DIQKDexp3}, further highlighted these challenges with no secure key extracted even under a restricted class of attacks.

The fact that optical systems have not yet succeeded in realizing a DIQKD protocol yet may seem surprising since they are generally the preferred experimental platform for QKD. Indeed, photonic circuits are an appealing choice for DIQKD for several reasons, including their ability to propagate information over long distance using telecom wavelengths, the availability of efficient commercial hardware, high repetition rate promising a potentially large key rate, and potential for circuit integration.

The exploration of photonic implementations of DIQKD is, in fact, a timely and active area of research. Several recent works have studied photonic setups for DIQKD, particularly in presence of a central station performing a Bell state measurement between the communicating parties. By heralding entanglement despite transmission losses, this extra station allows extending the achievable distance. Both the case of single photon sources~\cite{Gonzalez-Ruiz24,Twin-field_DIQKD} and SPDC sources~\cite{moradi25,alwehaibi25,ishihara25} have been investigated in this context. While promising for large distances, these schemes inherently come with an added complexity due to the need for a middle station and multiple photonic sources, which makes them harder to realize experimentally.

In this manuscript, we study the feasibility of DIQKD using a simple optical circuit identified for this purpose using reinforcement learning~\cite{valcarce_automated_2023}. We analyze the security of this circuit under realistic conditions of loss and noise using state-of-the art methods involving noisy pre-processing \cite{ho_noisy_2020,Woodhead_2021,Sekatski_2021} and the direct evaluation of the quantum conditional von Neumann entropy \cite{brown_device-independent_2021}. In turn, we introduce an efficient semidefinite programming (SDP) hierarchy for the evaluation of the conditional von Neumann entropy and extend finite-size secret key analysis to generic Bell scores. Together, these contributions allow us to significantly reduce the cost of finite-size effects for our circuit, and allow us to expect a key generation in a 10 hours experiment for a realistic setup efficiency $\eta=87.5\%$.

\section{Framework} \label{sec:framework}

\subsection{Protocol} \label{subsec:protocol}
Our DIQKD protocol follows the lines of Refs.~\cite{Acn2007,nadlinger_device-independent_2022}. Two distant parties, Alice and Bob, are connected by one classical and one quantum communication channel. Both channels are public and can be eavesdropped by an adversary called Eve. The classical channel (the internet for example) is authenticated, only allowing the adversary to listen to the communications without modifying them. Eve may however alter the signals on the quantum channel. Alice and Bob use the quantum channel to exchange $n$ quantum states one after the other, corresponding to $n$ distribution rounds.

After each system is distributed, Bob decides whether it belongs to a \textit{key generation round}, with probability $\gamma$, or a \textit{test rounds}, with probability $(1-\gamma)$, and announces it publicly on the classical channel. In the case of a key generation round, Bob chooses the setting $y=0$ and Alice chooses the settings $x=1$. They store their settings in local registers $X$ and $Y$. For test round, Bob chooses randomly between the settings $y\in\{1,2\}$, and Alice between the settings $x\in\{1,2\}$. These settings $x=1,2$ ($y=0,1,2$) are
implemented in Alice's (Bob's) measurement device by a corresponding physical parameter $\alpha_x$ ($\beta_y$) (in our case these parameters will correspond to choices of displacements, see below, in a polarization-based implementation it may be choices of angles on a polarization analyzer for instance). After setting their measurement choices, both parties measure their state accordingly, generating one bit each,
denoted $a$ and $b$, equal to 1 in case of photon detection and 0 otherwise. These results are stored in registers $A$ and $B$. During a key generation round, Alice further randomly switches her bit with probability $\mathrm{p}$, in order to implement noisy pre-processing \cite{ho_noisy_2020}.

After measurement of the $n$ quantum states, Alice and Bob hold the $n$ bits of their raw key $A^n$ and $B^n$ and proceed with classical processing.
First, Alice and Bob perform error correction as follows. Alice sends a syndrome $\mathrm{M}$, a small amount of bits computed from her raw key, publicly through the classical channel. Bob uses this message to attempt to reconstruct Alice’s raw key. If the retrieval of the key is successful they proceed, otherwise they abort.

The second step is privacy amplification. Now that Bob holds a copy of Alice's raw key, he estimates the Bell score $I$ and infers a bound on the smooth conditional min entropy $H_{\min}^{\epsilon_s}(A^n|X^nY^nE)$ using the entropy accumulation theorem \cite{dupuis_entropy_2020}. This allows him to compute the length $l$ of the key that can be guaranteed to be $\epsilon_\text{snd}$-close to a uniformly random string, i.e.~secret. Here, $\epsilon_\text{snd}$ quantifies the security of the final key and depends on various
parameters such as the smoothing parameter $\epsilon_s$. Finally, Alice and Bob extract the secret key by seeding a quantum-proof strong extractor with a short secret bit string and applying the corresponding hash function on Alice's raw key $A^n$.

\subsection{Key rate}

\paragraph{Asymptotic key rate}\label{subsec:asymptotic}
A complete security analysis bounds the length $l$ of the key that can be extracted by the protocol as a function of its parameters, such as the Bell score $I$, the number of rounds $n$ and the security parameter $\epsilon_\text{snd}$. As a first step towards that, it is useful to estimate the asymptotic key rate, which is the ratio $r=l/n$ between the length $l$ of the secret key that can be produced by a DI-QKD protocol and the number of rounds $n$, in the limit of large $n$. This quantity provides an upper bound on the actual finite key rate. In particular, if the asymptotic key rate vanishes, no finite size key can be extracted. 

Since we use one-way error correction, the asymptotic key rate is given by the Devetak and Winter formula \cite{devetak_distillation_2005}
\begin{equation}
    r_{DW} = H(A_1|E) - H(A_1|B_0),
    \label{eq:DW}
\end{equation}
where $A_1$ and $B_0$ are Alice and Bob's outcomes for a key generation round and $E$ stands for all the side-information available to the adversary. Here, the conditional von Neumann entropies $H(A_1|E)$ and $H(A_1|B_0)$ quantify the uncertainty that Eve and Bob respectively have on Alice's raw key bit $A_1$.

As shown in~\cite{ho_noisy_2020,Woodhead_2021,Sekatski_2021} the asymptotic key rate can be bounded from the Clauser-Horne-Shimoney-Holt (CHSH) score $S$~\cite{CHSH1969} and the pre-processing parameter $\mathrm{p}$ as $r_{DW} \geq r_0$ with
\begin{equation}\label{eq:Ho}
\begin{split}
    r_{0} &=  1 - h \left( \frac{1+\sqrt{ (S/2)^2 -1  }}{2} \right) \\
    &  + h \left( \frac{1 + \sqrt{1 - p(1-p)(8-S^2)}}{2} \right)- H(A_1|B_0).
\end{split}
\end{equation}
Since the CHSH score does not capture the full information about the statistics of a bipartite system, this bound is generally not optimal.

Remarkably, it was shown recently that the conditional von Neumann entropy $H(A_1|E)$ can be bounded from the knowledge of the complete probability distributions $\mathcal{P}(a,b|x,y)$~\cite{brown_device-independent_2021}, hence allowing for a direct computation of \cref{eq:DW}. This method formulates the problem as a double SDP hierarchy with parameters $\ell\in\mathbb{N}$ and $m\in\mathbb{N}$, whose result converges to $H(A_1|E)$ when $\ell, m \to \infty$. Here, $m$ is a parameter denoting the number of divisions of the segment $[0,1]$ used to bound the logarithm appearing in a formulation of the von Neumann entropy and $\ell$ corresponds to the level of the Navascués-Pironio-Acin (NPA) hierarchy~\cite{Navascues07}, namely the highest order of the polynomial of operators (of observable) taken into account. The higher these two parameters are, the closer the estimation is to the actual value, but the larger the semi-definite program is and the costlier the computation. For instance, when considering a bipartite scenario with two binary measurements per party, the SDP matrix is of size $(5+4m)^\ell\times(5+4m)^\ell$, which grows quickly with $m$. This directly impacts the SDP solver's memory requirement and solving time, effectively restricting the precision of the computation. This complexity has motivated the search for improved ways of computing lower bounds on the von Neumann entropy. For instance, Ref.~\cite{Kossmann25} was able to reduce the matrix size for the same scenario to $(5+2m)^\ell\times(5+2m)^\ell$, which still grows quickly with $m$.

However, this complexity is not necessary. In \cref{app:block}, we show how the von Neumann entropy $H(A_1|E)$ can be bounded with a converging SDP hierarchy parametrized by the same parameters $\ell$ and $m$, but whose SDP matrix size is independent of $m$. This hierarchy makes use of several SDP blocks, inducing a memory and computational cost that grows, for any $\ell$, only linearly with $m$. We refer to this SDP hierarchy as the \textit{block hierarchy}. Its improved scaling allows obtaining tighter lower bounds $H(A_1|E)_{\text{block},\ell,m}$ on $H(A_1|E)$ than the previous ones for the same computational effort. In this work, we thus lower bound the asymptotic key rate by
\begin{equation}
\label{eq:block_bound}
    r_{\text{block},\ell,m} = H(A_1|E)_{\text{block},\ell,m} - H(A_1|B_0).
\end{equation}
with the choice of $\ell$ and $m$ described in the result section.

\paragraph{Finite size analysis} \label{subsec:finitesize}
Since any experiment produces a finite amount of data, the asymptotic key rate $r_{DW}$ cannot certify the security of the key that can be produced in an actual experiment. Quantifying the key length that can be extracted in a practical and finite experiment involving $n$ rounds requires a dedicated finite size analysis~\cite{Scarani08}.

Recently, the Entropy Accumulation Theorem (EAT) was shown to be a successful framework for DI-QKD finite size analysis~\cite{dupuis_entropy_2020,ArnonFriedman2018}. In particular, this theorem allows proving security against the most general class of adversary, performing so-called \textit{coherent attacks}, while recovering the asymptotic key rate $r_{DW}$ when $n$ tends to infinity.

So far, finite size security proofs using the EAT have mostly based their security on the CHSH score as an entanglement estimator~\cite{ArnonFriedman2018,Murta19,tan_improved_2022,nadlinger_device-independent_2022,Hahn24,Twin-field_DIQKD}. But this Bell score may not always be tangent to the iso-surface $H(A_1|E) =$ cst for the statistics $\mathcal{P}(a,b|x,y)$ produced in the experiment, hence bounding $H(A_1|E)$ suboptimally. As mentioned earlier, improved entropy bounds are possible when considering the complete probability distribution $\mathcal{P}(a,b|x,y)$. We therefore develop a finite size DI-QKD analysis based on the EAT which directly depends on $\mathcal{P}(a,b|x,y)$. This is done by using a general $I$-score, which may differ from the CHSH score. Relevant $I$-scores tangent to the entropy isosurface can be obtained from the estimation of $H(A_1|E)_{\text{block},\ell,m}$ described above (see Appendix \ref{sec:Iscore} for more details). The details of our finite size analysis are provided in \cref{seq:finite_fullstat}. When the considered $I$-score differs strongly from CHSH, this can result in a significantly improved finite statistical analysis.

\subsection{Optical setup} \label{subsec:opticalsetup}
\begin{figure}
 \centering
 \includegraphics[width=0.95\linewidth]{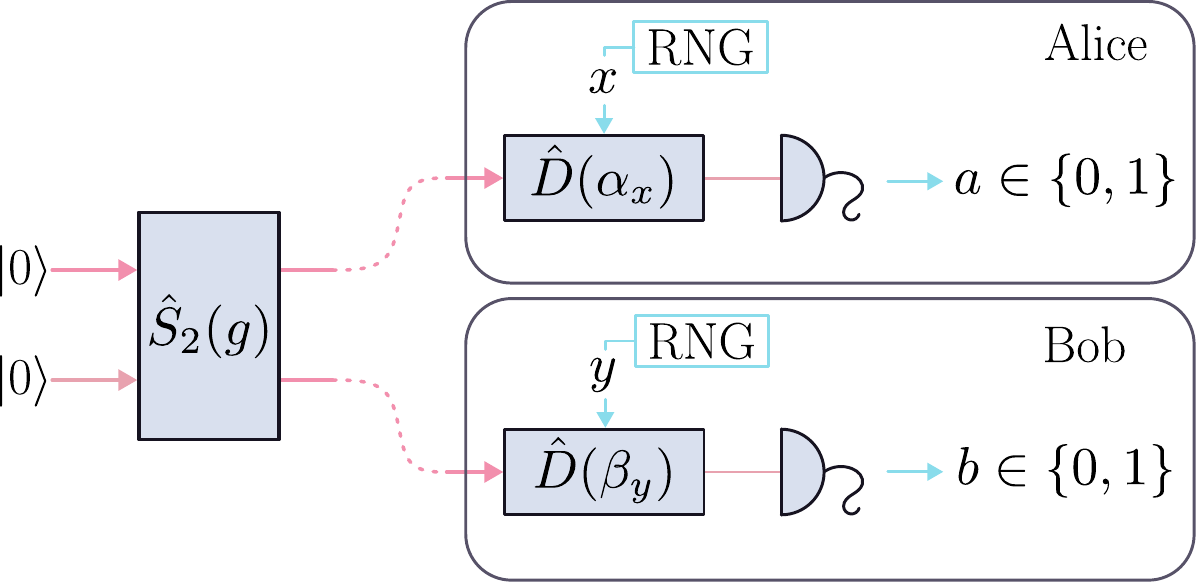}
 \caption{\textit{Optical setup.} Two vacuum modes are entangled via a two-mode squeezer and distributed to Alice and Bob. Each party uses a local random number generator (RNG) to select measurement settings ($x$ for Alice, $y$ for Bob). Their respective modes undergo displacements parameterized by $\alpha_x$ and $\beta_y$, subsequently measured via photon detectors, yielding binary outcomes $a$ and $b$ (0: no click, 1: click).}
\label{Optical setup studied}
\end{figure}

We consider the optical circuit represented in \cref{Optical setup studied}. This circuit is a slight simplification of the one obtained in Ref.~\cite{valcarce_automated_2023}, where reinforcement learning was used to identify optical circuits suited to DIQKD. Namely, squeezed-displacements have been replaced by simple displacements, with only a slight impact on the maximum CHSH score. An earlier study of this setup showed a high robustness to losses of its Bell violation, confirming its potential for the task of DIQKD~\cite{Brask12}.

The circuit involves a single frequency mode $\omega$ and two spatial modes: $(\hat{a}^{\dagger},\hat{a})$ for Alice and $(\hat{b}^{\dagger},\hat{b})$ for Bob, described by the Hilbert space
\begin{equation}
\begin{split}
\mathcal{H} &=\mathcal{H}_{\hat{a}^{\dagger},\hat{a}} \otimes \mathcal{H}_{\hat{b}^{\dagger},\hat{b}}.
\end{split}
\end{equation}

The circuit begins with a two-mode squeezer (with squeezing parameter $g$) that generates a two-mode optical state. The associated operator is 
\begin{equation}
    \hat{S}_2(g) = e^{g(\hat{a} \hat{b}-\hat{a}^{\dagger} \hat{b}^{\dagger})/2}.
    \label{eq:TMS_exp}
\end{equation}
Two-mode squeezers are known to generate a highly correlated state. One mode is sent to Alice and the other one to Bob. Both of them then apply a displacement operation. The associated operator, with displacement parameter $\alpha$ for Alice is
\begin{equation}
    \hat{D}(\alpha) = e^{\alpha\hat{a}^{\dagger}-\alpha^{*}\hat{a}}. \label{eq:Displacement_exp}
\end{equation}
Similarly, the displacement operator for Bob is $\hat{D}(\beta) = e^{\beta\hat{b}^{\dagger}-\beta^{*}\hat{b}}$.
Finally, the state is detected by means of non-photon-number-resolving detectors with efficiency $\eta$. We note that this efficiency accounts for all the losses that can occur in the two-mode squeezer, the optical fibers or in the measurements, since they can be pulled back into the efficiency of the detectors. We assign the result $0$ to a no-click event and a $1$ to a click event of the detectors. In \cref{seq:proba}, we derive the probability distribution $\mathcal{P}(a,b|x,y)$ associated with this optical circuit as function of the circuit parameters $g$, $\eta$, $\alpha$ and $\beta$.

This circuit involves only few linear optical elements and two single photon detectors. Importantly, these components have known implementations with commercially available hardware at telecom wavelength that are able to operate at a high repetition rate.

Apart from losses, the main source of noise in this optical circuit comes from the leakage of optical components into unintended optical modes. For instance, this occurs if the displacement operator is applied to an optical mode that is different from the two-mode squeezer. We provide a modelization of this noise and discuss its impact on the asymptotic key rate in \cref{seq:NoiseModels}.

\section{Results} \label{sec:results}
We first compute the asymptotic DIQKD key rate for this photonic circuit as a function of the efficiency $\eta$. For comparison, we do so both as a function of the CHSH score and of the full statistics.

In the first case, we maximize $r_0$ over the squeezing parameter $g$, Alice's displacement parameters $\alpha_1, \alpha_2$, Bob's displacement parameters $\beta_0, \beta_1, \beta_2$, and the noisy pre-processing parameter $\mathrm{p}$. This optimization is performed using standard gradient-descent and Nelder-Mead routines~\cite{nelder65}. In the second case, we maximize $r_{\mathrm{block},\ell,m}$ over the same parameters. Since this function involves solving an SDP, it is more costly to evaluate than $r_0$, so we adopt a custom optimization method, detailed in \cref{sec:I_score}. 
The SDP matrix is constructed with the monomials $\mathcal{O} = \mathcal{O}_{AB}\times \mathcal{O}_Z$, where $\mathcal{O}_{AB}=[[\hat{\mathbb{I}}, A_1, A_2]\times[\hat{\mathbb{I}}, B_1, B_2], A_1 A_2, A_2 A_1, A_1 B_2 B_1, [A_1 A_2 + A_2 A_1] \times B_1 B_2]$, $\mathcal{O}_Z=[\hat{\mathbb{I}},Z_1,Z_1^\dag,Z_2,Z_2^\dag]$ and we use the parameter $m=12$.

The result of these two optimizations is shown in \cref{fig:asymptotic_keyrate}. We note the significant improvement on the critical efficiency using this method. 
Specifically, using the full statistics to compute the key rate reduces by $\approx4\%$ the requirement on efficiency to obtain a key rate higher than $10^{-5}$.
At higher efficiency, the key rate $r_{\mathrm{block},\ell,m}$ can be up to one order  of magnitude greater than $r_0$.

\begin{figure}
  \centering
  \includesvg[width=\linewidth]{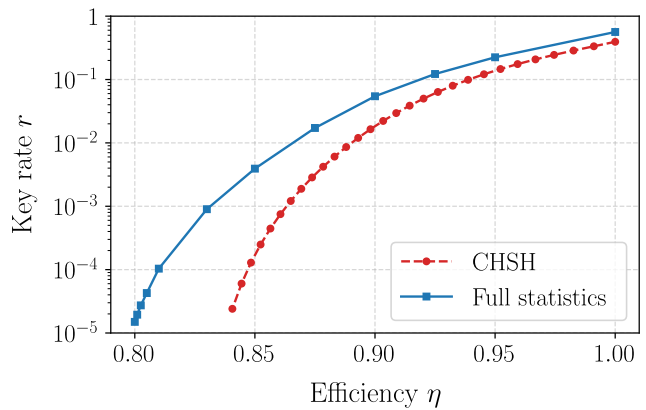}
  \caption{Lower bound on the asymptotic key rate as a function of the efficiency $\eta$, with a security based on the CHSH score (in red), and on the full statistics (in blue).}
  \label{fig:asymptotic_keyrate}
\end{figure}

In order to assess the feasibility of an experiment, we now consider the cost for finite statistics. Namely, we compute the minimum number of rounds that allows producing a key with a positive length $l>0$ as a function of the efficiency $\eta$. Again, we compare the case of a security based on the CHSH score and one using the full statistics. The first case is computed using the method provided in Ref.~\cite{nadlinger_device-independent_2022}. The second case uses our new analysis detailed in \cref{seq:finite_fullstat}. The result is shown in \cref{fig:finite_size}.

\begin{figure}
  \centering
  \includesvg[width=\linewidth]{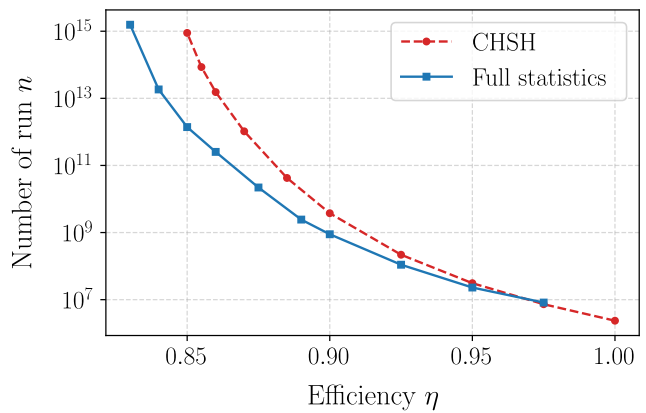}
  \caption{Minimum number of run $n$ leading to a positive key rate with $\epsilon_\text{snd}=3\times 10^{-10}$ when basing the security on the CHSH score (in red), and on full statistics (in blue), as a function of the efficiency $\eta$.}
  \label{fig:finite_size}
\end{figure}

In this figure, we see that using the full statistics significantly reduces the number of experimental runs required to extract a secret key, potentially by several orders of magnitude. This improvement, particularly for efficiencies close to the threshold, can render an otherwise impossible experiment practically feasible. Concretely, for a repetition rate of 1 MHz, our result shows that an efficiency of $\eta=87.5\%$ is sufficient to produce a secret key within approximately 8 hours, corresponding to $n=3\times 10^{10}$ rounds. This is achieved with the squeezing
$T_g=0.249$ or $2.163$dB, pre-processing parameter $\mathrm{p}=0.042$ and displacements $\alpha_1=0.024$, $\alpha_2=-0.521$, $\beta_0=0.013$, $\beta_1=-0.104$ and $\beta_2=0.0340$.

\section{Conclusion} \label{sec:conclusion}
In summary, we presented a comprehensive feasibility assessment for implementing DIQKD with a simple optical circuit, demonstrating that a photonic realization of DIQKD is within experimental reach. This was made possible by developing an efficient method for evaluating the conditional von Neumann entropy and a finite-size DIQKD security analysis that fully exploits the available statistics. As a result, we reduced the number of experimental trials required to obtain a secret key by several orders of magnitude, thereby identifying a feasibility region for DIQKD using our photonic circuit. This opens the way for a first optical demonstration of DIQKD with commercial equipment, marking a significant step toward practical device-independent secure communication.

The methods developed here extend well beyond the specific setup considered. The converging block SDP method also applies to the hierarchy introduced in~\cite{Kossmann25} for bounding the von Neumann entropy, as well as to the computation of additional entropic quantities, such as the sandwiched Rényi entropies $H_\alpha^{\uparrow}(A|E)$~\cite{Hahn24}. These generalizations may, in turn, further enhance finite key length analyses via the Rényi-EAT framework~\cite{Dupuis23,Arqand24}. Likewise, our finite-size analysis based on full statistics offers a versatile approach to improve security analyses in alternative photonic scenarios, including spontaneous parametric down-conversion setups and architectures employing on a central station. Finally, emerging proposals based on routed Bell tests \cite{LeRoyDeloison25,sekatski25,Tan24} represent another promising direction where our techniques could strengthen quantitative feasibility benchmarks.

\begin{acknowledgments}
We thank Peter Brown, Omar Fawzi, Igor Klep, Marc-Olivier Renou, Nicolas Sangouard, Xiangling Xu and Yuzhe Zhang for discussions. We thank Erik Woodhead for the Julia library ``QUantumNPA''. C.~L, X.~V. and J-D.~B acknowledge funding by the European Union’s Horizon Europe research and innovation programme under the project “Quantum Secure Networks Partnership” (QSNP, Grant Agreement No. 101114043) and by a French national quantum initiative managed by Agence Nationale de la Recherche in the framework of France 2030 with the reference ANR-22-PETQ-0009.
\end{acknowledgments}

\bibliography{references}

\newpage

\onecolumngrid

\newpage

\appendix

\section{Optical setup}
This annex is devoted to the optical setup discussed in the main text. It describes the probability distribution of the setup with its derivation and a modelisation for the leaking of the optical components (two-mode squeezer and displacement) into other optical modes.

\subsection{Computation of the probabilities} \label{seq:proba}

\subsubsection{Formula sheet}

We will use these formulas for the derivation of the probability distributions of the circuit.

Displacement operator :

\begin{equation}
    \hat{D}(\alpha) = e^{\alpha\hat{a}^{\dagger}-\alpha^{*}\hat{a}}
\end{equation}

Two-mode squeezer operator :

\begin{equation}
    \hat{S}_2(r) = e^{r(\hat{a}_i\hat{a}_j-\hat{a}_i^{\dagger}\hat{a}_j^{\dagger})/2}
\end{equation}

Coherent state :

\begin{equation}
    \ket{\alpha} = \hat{D}(\alpha) \ket{0} 
\end{equation}

Commutation of creation-annihilation operator with a function :

\begin{equation}
\label{eqn:1}
    x^{\hat{a}^{\dagger} \hat{a}} f(\hat{a}^{\dagger}) = f(x \hat{a}^{\dagger}) x^{\hat{a}^{\dagger} \hat{a}}
\end{equation}

Two-mode squeezer applied on vacuum (with $T_g=\tanh(g)$) :

\begin{equation}
\label{eqn:2}
    \hat{S}_2(r) \ket{00} = \sqrt{1-T_g^2} e^{T_g \hat{a}^{\dagger} \hat{b}^{\dagger}} \ket{00}
\end{equation}

Application of a two-mode squeezer on coherent states : 

\begin{equation}
\label{eqn:3}
    e^{T_g \hat{a} \hat{b}} \ket{\alpha \beta} = e^{T_g \alpha \beta} \ket{\alpha \beta}
\end{equation}

Displacement operator applied on the vacuum :

\begin{equation}
\label{eqn:4}
    \hat{D}(\alpha) \ket{0} = e^{-\frac{\abs{\alpha}^2}{2}} e^{\alpha \hat{a}^{\dagger}} \ket{0} 
\end{equation}

Inverse and adjoint of the displacement operator :

\begin{equation}
\label{eqn:5}
   \hat{D}^{\dagger}(\alpha) = \hat{D}^{-1}(\alpha) = \hat{D}(-\alpha)
\end{equation}

Composition of two displacement operator :

\begin{equation}
\label{eqn:6}
    \hat{D}(\alpha) \hat{D}(\beta) = e^{\frac{1}{2}(\alpha \beta^* - \alpha^*\beta)} \hat{D}(\alpha + \beta)
\end{equation}

Scalar product of two coherent states :

\begin{equation}
\label{eqn:7}
    \innerproduct{\alpha}{\beta}= e^{-\frac{1}{2}( \abs{\beta}^2 + \abs{\alpha}^2 - 2 \alpha^* \beta)}
\end{equation}

\subsubsection{Probability distribution of the circuit}

In this section we compute the joint and marginal probabilitie $\mathcal{P}(00|xy)$, $\mathcal{P}_A(0|x)$ and $\mathcal{P}_b(0|y)$ for the photonic circuit described in the main text. Note that the remaining probabilities can be retrieved using the constraints $\mathcal{P}_A(a|x) = \mathcal{P}(ab|xy) + \mathcal{P}(a\bar{b}|xy)$, $\mathcal{P}_B(b|y) = \mathcal{P}(ab|xy) + \mathcal{P}(\bar{a}b|xy) $ and the sum of probability being equal to 1. Here, the settings $x$ and $y$ correspond to displacements $\alpha_x$ and $\beta_y$. In the remaining of this appendix, we parametrize the statistics as a function of these displacements, which we label simply $\alpha$ and $\beta$.

We want to compute the probability of having two simultaneous no-click events on an optical circuit that generates the state
\begin{equation}
\ket{\psi} =  \hat{D}_a(\alpha) \hat{D}_b(\beta) \hat{S}_2(r) \ket{00}.
\end{equation}

We modelize the detector's no-click by the POVM $R = 1 - \eta$ with $\eta$ the efficiency of the detectors. Notice that for the $n$-photon Fock state $\hat{\rho} = \ket{n}\bra{n}$, this produces a no-click probability $\mathcal{P}(0)= \Tr(\hat{\rho} R) = (1-\eta)^n$ as expected.

We thus want to compute
\begin{align}
    \mathcal{P}(00|\alpha,\beta) &= Tr(\hat{\rho} R^{\hat{a}^{\dagger} \hat{a} + \hat{b}^{\dagger} \hat{b}} ) = \bra{\psi} R^{\hat{a}^{\dagger} \hat{a} + \hat{b}^{\dagger} \hat{b}}  \ket{\psi}\\
    &= \bra{00} \hat{S}_2^{\dagger}(r) \hat{D}_a^{\dagger}(\alpha) \hat{D}_b^{\dagger}(\beta)
    R^{\hat{a}^{\dagger} \hat{a} + \hat{b}^{\dagger} \hat{b}} \hat{D}_a(\alpha) \hat{D}_b(\beta) \hat{S}_2(r) \ket{00}
\end{align}

We can put it in a symmetric form

\begin{equation}
    \mathcal{P}(00|\alpha,\beta) = \bra{00} \hat{S}_2^{\dagger}(r) \hat{D}_a^{\dagger}(\alpha) \hat{D}_b^{\dagger}(\beta)
    R^{\frac{\hat{a}^{\dagger} \hat{a} + \hat{b}^{\dagger} \hat{b}}{2}} 
    R^{\frac{\hat{a}^{\dagger} \hat{a} + \hat{b}^{\dagger} \hat{b}}{2}} 
    \hat{D}_a(\alpha) \hat{D}_b(\beta) \hat{S}_2(r) \ket{00}
\end{equation}

and insert the identity  ( $\mathbb{I}_d = \frac{1}{\pi^2} \int_{\mathbb{C}} \int_{\mathbb{C}} \ket{\gamma \gamma'} \bra{\gamma \gamma'} d\gamma d\gamma'$ ):

\begin{equation}
    \mathcal{P}(00|\alpha,\beta) =      \frac{1}{\pi^2} \int_{\mathbb{C}} \int_{\mathbb{C}}
      \abs{
       \bra{00} \hat{S}_2^{\dagger}(r) \hat{D}_a^{\dagger}(\alpha) \hat{D}_b^{\dagger}(\beta)
    R^{\frac{\hat{a}^{\dagger} \hat{a} + \hat{b}^{\dagger} \hat{b}}{2}} 
    \ket{\gamma \gamma'} }^2d\gamma d\gamma'.
\end{equation}

So firstly we have to compute $ \bra{00} \hat{S}_2^{\dagger}(r) \hat{D}_a^{\dagger}(\alpha) \hat{D}_b^{\dagger}(\beta)
    R^{\frac{\hat{a}^{\dagger} \hat{a} + \hat{b}^{\dagger} \hat{b}}{2}} 
    \ket{\gamma \gamma'} $.

\subsubsection{Computation}

Firstly we apply the POVM associated to a no click to coherent states

\begin{subequations}
\begin{align} \label{eq1}
R^{\frac{1}{2}(\hat{a}^{\dagger} \hat{a} + \hat{b}^{\dagger} \hat{b})} \ket{\gamma \gamma'} & = R^{\frac{1}{2}(\hat{a}^{\dagger} \hat{a} + \hat{b}^{\dagger} \hat{b})} \hat{D}_1(\gamma) \hat{D}_2(\gamma')
     \ket{00} \\
 & =_{(\ref{eqn:4})} R^{\frac{1}{2}(\hat{a}^{\dagger} \hat{a} + \hat{b}^{\dagger} \hat{b})} e^{-\frac{1}{2} ( \abs{\gamma}^2 + \abs{\gamma'}^2   )} e^{\gamma \hat{a}^{\dagger}} e^{\gamma' \hat{b}^{\dagger} } \ket{00} \\
 &= e^{-\frac{1}{2} ( \abs{\gamma}^2 + \abs{\gamma'}^2   )} (R^{\frac{1}{2}(\hat{a}^{\dagger} \hat{a})} e^{\gamma \hat{a}^{\dagger}})  (R^{\frac{1}{2}(\hat{b}^{\dagger} \hat{b})} e^{\gamma' \hat{b}^{\dagger} }) \ket{00} \\
 &=_{(\ref{eqn:1})} e^{-\frac{1}{2} ( \abs{\gamma}^2 + \abs{\gamma'}^2   )} ( e^{\sqrt{R} \gamma \hat{a}^{\dagger}} R^{\frac{1}{2}(\hat{a}^{\dagger} \hat{a})} ) ( e^{\sqrt{R} \gamma' \hat{b}^{\dagger}} R^{\frac{1}{2}(\hat{b}^{\dagger} \hat{b})} ) \ket{00} \\
 & =_{(\ref{eqn:4})} e^{-\frac{1}{2} ( \abs{\gamma}^2 + \abs{\gamma'}^2   )} e^{ \frac{R}{2} ( \abs{\gamma}^2 + \abs{\gamma'}^2   )} \hat{D}_a(\sqrt{R} \gamma) \hat{D}_b(\sqrt{R} \gamma') \ket{00}.
\end{align}
\end{subequations}

We are left with two new displacement operators on vacuum. Then we compose these displacement with the one in the initial formula

\begin{equation} \label{eq2}
    \begin{split}
    \hat{D}_b^{\dagger}(\beta) \hat{D}_a^{\dagger}(\alpha) \hat{D}_a(\sqrt{R} \gamma) \hat{D}_b(\sqrt{R} \gamma') \ket{00} &=_{(\ref{eqn:5})} \hat{D}_b(-\beta) \hat{D}_a(-\alpha) \hat{D}_a(\sqrt{R} \gamma) \hat{D}_b(\sqrt{R} \gamma') \ket{00} \\
    &=_{(\ref{eqn:6})} e^{\frac{1}{2}( -\alpha \sqrt{R} \gamma^* + \alpha^* \sqrt{R} \gamma)} e^{\frac{1}{2}( -\beta \sqrt{R} (\gamma')^* + (\beta)^* \sqrt{R} \gamma')} \ket{\sqrt{R}\gamma-\alpha, \sqrt{R}\gamma' - \beta}.
    \end{split}
\end{equation}

Now we apply the two-mode squeezer on the new displacement operators

\begin{equation} \label{eq3}
\begin{split}
\bra{00} \hat{S}_2^{\dagger}(r) \ket{\sqrt{R}\gamma-\alpha, \sqrt{R}\gamma' - \beta} &=_{(\ref{eqn:2})} 
\bra{00} \sqrt{1-T_g^2} e^{T_g \hat{a} \hat{b}} \ket{\sqrt{R}\gamma-\alpha, \sqrt{R}\gamma' - \beta} \\
&=_{(\ref{eqn:3})} \bra{00} \sqrt{1-T_g^2} e^{T_g (\sqrt{R}\gamma-\alpha)(\sqrt{R}\gamma'-\beta)} \ket{\sqrt{R}\gamma-\alpha, \sqrt{R}\gamma' - \beta}.
\end{split}
\end{equation}

And lastly we compute the dot product of the remaining ket with the bra of the vacuum

\begin{equation}
    \bra{00} \ket{\sqrt{R}\gamma-\alpha, \sqrt{R}\gamma' - \beta} =_{(\ref{eqn:7})} e^{\frac{-1}{2} (\abs{\sqrt{R}\gamma-\alpha}^2 + \abs{\sqrt{R}\gamma'-\beta}^2 )}.
\end{equation}

\subsubsection{Result}

We compose all above steps in the initial formula 

\begin{equation}
\begin{split}
   \sqrt{1-T_g^2} e^{ \frac{(R-1)}{2} ( \abs{\gamma}^2 + \abs{\gamma'}^2   )} e^{\frac{\sqrt{R}}{2}( -\alpha  \gamma^* + \alpha^*  \gamma)} e^{\frac{\sqrt{R}}{2}( -\beta (\gamma')^* + (\beta)^* \gamma')} e^{T_g (\sqrt{R}\gamma-\alpha)(\sqrt{R}\gamma'-\beta)} e^{-\frac{1}{2} (\abs{\sqrt{R}\gamma-\alpha}^2 + \abs{\sqrt{R}\gamma'-\beta}^2 )}
\end{split}
\end{equation}

and we can now simplify this expression

\begin{equation}
\begin{split}
   \sqrt{1-T_g^2} e^{ -\frac{1}{2}( \abs{\gamma}^2 + \abs{\gamma'}^2) + \sqrt{R}( \alpha^* \gamma + (\beta)^*(\gamma')) - \frac{1}{2}(\abs{\alpha}^2 + \abs{\beta}^2) + T_g(\gamma \sqrt{R} - \alpha)(\gamma' \sqrt{R} - \beta) }
\end{split}
\end{equation}

Integrating now over $\gamma,\gamma' \in \mathbb{C}$ and taking the modulus square of this expression (using a formal expression program), we finally obtain:

\begin{subequations}
\begin{align}
    \mathcal{P} ( 00 | \alpha, \beta ) &= \frac{1-T_g^2}{1-R^2 T_g^2} e^{ (|\alpha|^2 + |\beta|^2) \frac{1 - R + R^2T_g^2 - RT_g^2}{R^2 T_g^2 -1} - \frac{2 \alpha \beta T_g (R-1)^2}{R^2T_g^2 -1} }\\
    \mathcal{P} ( 0 | \alpha ) &= \frac{1-T_g^2}{1-R T_g^2} e^{ |\alpha|^2  \frac{1 - R + R^2T_g^2 - RT_g^2}{R T_g^2 -1} }\\
    \mathcal{P} ( 0 | \beta ) &= \frac{1-T_g^2}{1-R T_g^2} e^{ |\beta|^2  \frac{1 - R + R^2T_g^2 - RT_g^2}{R T_g^2 -1} }.
\end{align}
\end{subequations}

The asymptotic key rate obtained with these statistics is described in the main text.

\subsection{Noise models}
\label{seq:NoiseModels}

In this appendix, we take a look at a noise model for the optical circuit studied. The idea is to take into account the possibility that different optical parts could excite different optical mode (in frequency for example).

Initially, we have one optical mode $(\hat{a}^{\dagger},\hat{a})$ for Alice and another one $(\hat{b}^{\dagger},\hat{b})$ for Bob. 
\begin{equation*}
\begin{split}
\mathcal{H} &= \mathcal{H}_{\omega}\\
&=\mathcal{H}_{\hat{a}^{\dagger},\hat{a}} \otimes \mathcal{H}_{\hat{b}^{\dagger},\hat{b}}.
\end{split}
\end{equation*}
In the noise model, we split the modes in two : $(\hat{a}^{\dagger}_1,\hat{a}_1)$ and $(\hat{a}^{\dagger}_2,\hat{a}_2)$ for Alice, $(\hat{b}^{\dagger}_1,\hat{b}_1)$ and $(\hat{b}^{\dagger}_2,\hat{b}_2)$ for Bob. The Hilbert space can thus be written
\begin{equation*}
\begin{split}
\mathcal{H} &= \mathcal{H}_{\omega_1} \otimes \mathcal{H}_{\omega_2}\\
&= (\mathcal{H}_{\hat{a}^{\dagger}_1,\hat{a}_1} \otimes \mathcal{H}_{\hat{b}^{\dagger}_1,\hat{b}_1} ) \otimes (\mathcal{H}_{\hat{a}^{\dagger}_2,\hat{a}_2} \otimes \mathcal{H}_{\hat{b}^{\dagger}_2,\hat{b}_2}).
\end{split}
\end{equation*}

First, we model a displacement operator that leaks to another optical mode. Taking $\hat{c}\in\{\hat{a},\hat{b}\}$, the operator is
\begin{equation*}
\begin{split}
    \hat{D}_c(\alpha,\chi) &= e^{ \sqrt{\chi}(\alpha \hat{c}_1^{\dagger} - \alpha^* \hat{c}_1) + \sqrt{1-\chi}(\alpha \hat{c}_2^{\dagger} - \alpha^* \hat{c}_2) } \\
     &= (e^{ \sqrt{\chi}(\alpha \hat{c}_1^{\dagger} - \alpha^* \hat{c}_1)}) \otimes (e^{  \sqrt{1-\chi}(\alpha \hat{c}_2^{\dagger} - \alpha^* \hat{c}_2) }) \\
     &= \hat{D}_{c}(\sqrt{\chi}\alpha) \otimes \hat{D}_{c}(\sqrt{1-\chi}\alpha)
\end{split}
\end{equation*}

This operator is a displacement operator that c (Alice or Bob) apply to their two optical modes with a parameter $\chi \in [ 0,1 ]$. In the limit $\chi=1$, only the first optical mode is displaced, and in the other limit $\chi=0$, only the second optical mode is displaced. We also note that the mean number of photons is conserved: with $\alpha_1=\sqrt{\chi}\alpha $ and $×\alpha_2 =\sqrt{1-\chi}\alpha $, we have $|\alpha_1|^2+|\alpha_2|^2 = |\alpha|^2$.
\newline

Second, we model a two mode squeezer that leaks into other optical modes with parameter $\zeta\in [0,1]$:
\begin{equation*}
\begin{split}
    \hat{S}_2(g,\zeta) \ket{0000} &= \sqrt{1-\zeta T_g^2} \sqrt{1-(1-\zeta)T_g^2} e^{\sqrt{\zeta}T_g \hat{a}_1^{\dagger} \hat{b}_1^{\dagger}  + \sqrt{1-\zeta}T_g \hat{a}_2^{\dagger} \hat{b}_2^{\dagger}} \ket{0000} \\
    &= ( \sqrt{1-\zeta T_g^2}  e^{\sqrt{\zeta}T_g \hat{a}_1^{\dagger} \hat{b}_1^{\dagger}  }     \ket{00}) \otimes (  \sqrt{1-(1-\zeta)T_g^2} e^{ \sqrt{1-\zeta}T_g \hat{a}_2^{\dagger} \hat{b}_2^{\dagger}}     \ket{00} ) \\
    &= (\hat{S}_2(g') \ket{00}) \otimes (\hat{S}_2(g") \ket{00}),
\end{split}
\end{equation*}
where $g = \text{arcth}(T_g )$ and we define $g' = \text{arcth}( \sqrt{\chi} T_g )$ and $g" = \text{arcth}( \sqrt{1-\chi} T_g )$.

In the limit $\zeta=1$, only the first optical mode is two mode squeezed, and the other limit $\zeta=0$ only the second mode is two mode squeezed). Again, this operator conserves the photon number: with $T_{g1} =\sqrt{\zeta}T_g $ and $T_{g2} =\sqrt{1-\zeta}T_g  $, we have $T_{g1}^2+T_{g2}^2 = T_g^2$ so we keep the mean number of photons. So we can define the second operator applied on vacuum
\newline

We now use these two operators $\hat D_c(\alpha,\chi)$ and $\hat S_2(g,\zeta)$ to model noise which correspond to the excitation of another optical mode by the quantum operations (TMS and displacement) in two different ways.
\newline

\paragraph{Two mode squeezing on two optical modes}    

First, we consider the case when the TMS creates photons in two different optical modes and the displacement operator acts only on one of the two optical modes. This situation is illustrated in \cref{fig:noise1a}.

\begin{figure}
    \centering
    \includegraphics[width=0.7\linewidth]{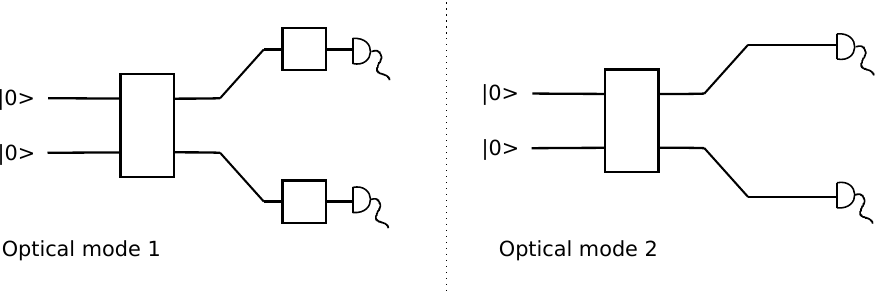}
    \caption{In the first noise model, the TMS acts on two families of optical modes, but the displacements only act on one of them.}
    \label{fig:noise1a}
\end{figure}

The state obtained is

\begin{subequations}
\begin{align}
\ket{\psi} &= 
    \hat{D}_a(\alpha,\chi = 1) \hat{D}_b(\alpha',\chi = 1) \hat{S}_2(g,\zeta) \ket{0000} \\ 
    &= (\hat{D}_a(\alpha)\otimes \mathbb{I}_d) ( \hat{D}_b(\alpha') \otimes \mathbb{I}_d)\left( (\hat{S}_2(g') \ket{00}) \otimes (\hat{S}_2(g") \ket{00}) \right) \\
    &=(  \hat{D}_a(\alpha) \hat{D}_b(\alpha') \hat{S}_2(g')    \ket{00}) \otimes  (\hat{S}_2(g")     \ket{00} ).
\end{align}
\end{subequations}

Here, the parameter $\zeta \in [0,1]$ modelizes the leaking in the second optical mode ($\zeta = 1.0$ meaning no leaking).
Here, we defined $g' = \artanh( \sqrt{\zeta} \sinh(g) )$ and $g" = \artanh( \sqrt{1-\zeta} \sinh(g) )$ so that the mean photon number $\sinh(g)^2 = \sinh(g')^2 + \sinh(g")^2$ is conserved.

\paragraph{Displacement on two optical modes}

In the second case, we suppose the TMS operation is applied on only one optical mode and the displacement operations of Bob and Alice are applied on two optical modes, see \cref{fig:noise2a}.

\begin{figure}
    \centering
    \includegraphics[width=0.7\linewidth]{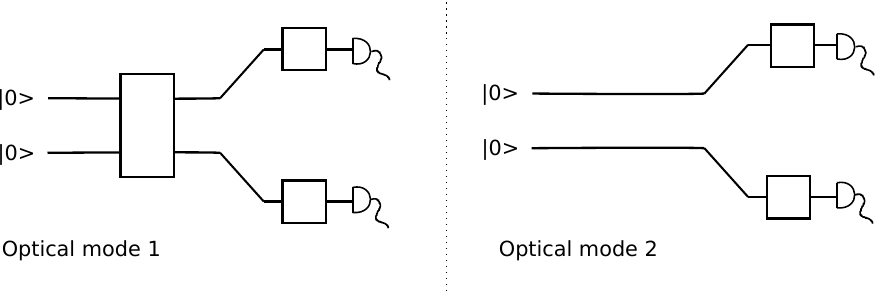}
    \caption{In the second noise model, the displacements act on two families of optical modes, but the TMS only acts on one of them.}
    \label{fig:noise2a}
\end{figure}

The state obtained is

\begin{subequations}
\begin{align}
\ket{\psi} &= 
    \hat{D}_a(\alpha,\chi) \hat{D}_b(\alpha',\chi) \hat{S}_2(g,\zeta=1) \ket{0000} \\ &= (\hat{D}_a(\sqrt{\chi}\alpha) \otimes \hat{D}_b(\sqrt{1-\chi}\alpha)) (\hat{D}_a(\sqrt{\chi}\alpha) \otimes \hat{D}_b(\sqrt{1-\chi}\alpha))   ( \hat{S}_2(g) \ket{00} )\otimes ( \ket{00}) \\
    &= (\hat{D}_a(\sqrt{\chi}\alpha) \hat{D}_b(\sqrt{\chi}\alpha') \hat{S}_2(g)\ket{00}) \otimes (\hat{D}_a(\sqrt{1-\chi}\alpha) \hat{D}_b(\sqrt{1-\chi}\alpha') \ket{00})
\end{align}
\end{subequations}

Here, the parameter $\zeta \in [0,1]$ modelizes the leaking in the second optical mode ($\zeta = 1.0$ meaning no leaking).

\paragraph{Results} 

The probability distribution associated to these new state are easily derived, given two optical modes, the probability that there is no overall click is equal to the probability that the first mode doesn't click times the probability that the second mode doesn't click as well. This allows us to compute the asymptotic key rate $r_0$ in presence of noise. The results are shown in \cref{fig:Noise_model_1,fig:Noise_model_2}.

In \cref{fig:Noise_model_1}, we can see the impact on the key rate of the parameter $\zeta$ from the leaking of the two-mode squeezer in another optical mode. In \cref{fig:Noise_model_2} we can see the impact of the parameter $\chi$ from the leaking of the displacement operators in another optical mode. We observe that the key rate is robust to both types of leakages.

\begin{figure}
 \centering
 \includegraphics[width=0.5\linewidth]{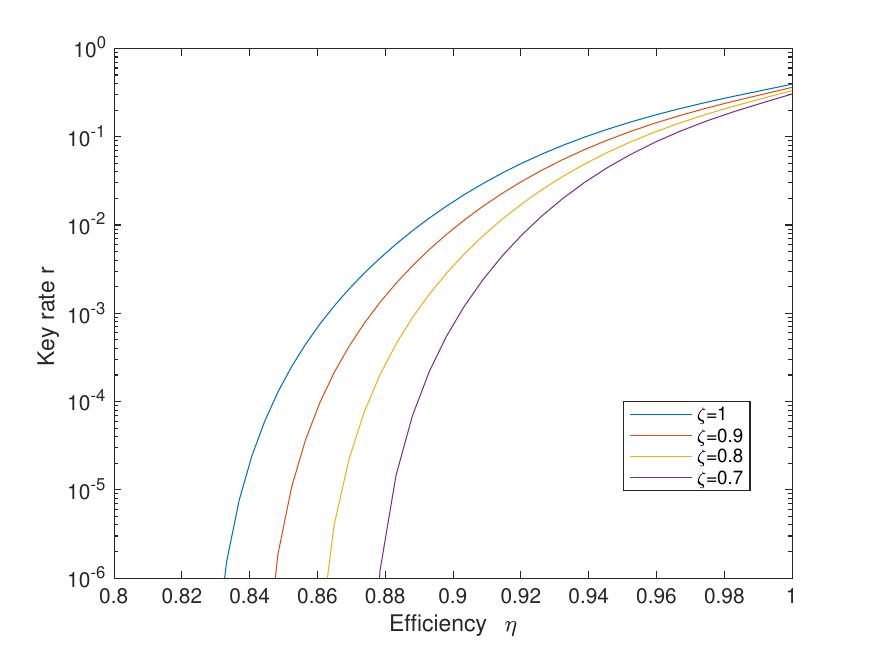}
 \caption{Lower bound of $r_0$($\eta$) with different values of the parameter $\zeta$}
 \label{fig:Noise_model_1}
\end{figure}

\begin{figure}
 \centering
 \includegraphics[width=0.5\linewidth]{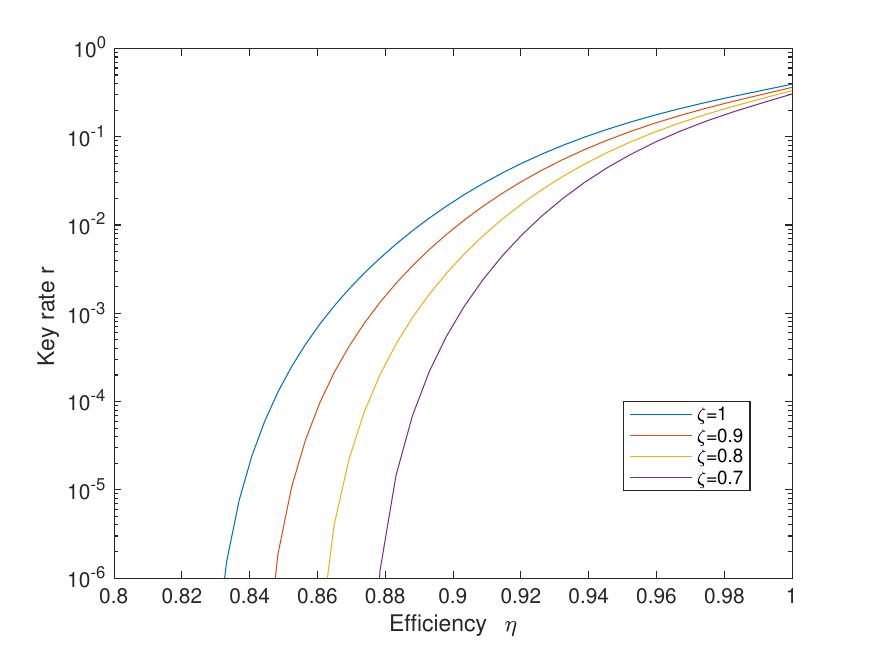}
  \caption{Lower bound of $r_0$($\eta$) with different values of the parameter $\chi$}
 \label{fig:Noise_model_2}
\end{figure}

\section{Block SDP hierarchy}\label{app:block}
Here, we recall the two SDP hierarchies introduced in~\cite{brown_device-independent_2021} to bound the conditional von Neumann entropy $H(A|E)$, which we refer to as the \textit{full} and \textit{split} hierarchies. We then propose a new SDP hierarchy to bound the same quantity that we refer to as the \textit{block} hierarchy, prove its convergence and compare it to the two previous ones.
\subsection{Definition of the \textit{full}, \textit{split} and \textit{block} hierarchies}
\label{seq:SDPchoice} 

In a scenario where three parties Alice, Bob and Eve share a quantum state $\rho_{ABE}$, Ref.~\cite{brown_device-independent_2021} showed that the von Neumann entropy $H(A|E)$ of the outcome of a binary observable $\hat A$ conditioned on a quantum register $E$ can be bounded by a series of polynomial optimizations of the form
\begin{align}\label{eq:polopt1}
\sum_{i=1}^m \frac{\omega_i}{t_i \ln(2)} \sum_{a=0}^1 \inf_{\hat Z_{i,a}} \Tr\left[ \rho_{AE} \left( \hat A^{(a)} \otimes (\hat Z_{i,a}+\hat Z_{i,a}^\dag + (1-t_i)\hat Z_{i,a}^\dag \hat Z_{i,a}) + t_i(\mathbb{I}_A\otimes \hat Z_{i,a}\hat Z^\dag_{i,a})\right)\right],
\end{align}
where $m\in\mathbb{N}$, $\omega_i$ and $t_i$ are the weights and nodes of a Gauss-Radau quadrature, $(\hat A^{(a)})^2=\hat A^{(a)}$ are the projectors of the observable $\hat A=\hat A^{(0)}-\hat A^{(1)}$, $\hat Z_{i,a}$ are non-hermitian operators acting on Eve's Hilbert space and the dagger denotes hermitian conjugation.

When only the behavior $\mathcal{P}(a,b|x,y)$ for the Alice-Bob system is known, the von Neuman entropy $H(A|E)$ can be bounded by minimizing \cref{eq:polopt1} over all quantum states $\rho_{ABE}$ and measurements $\hat A_x$, $\hat B_y$, in Hilbert spaces of arbitrary dimension, that are compatible with the behavior. When $\hat A=\hat A_1$, this leads to the following polynomial optimization:
\begin{align}\label{eq:polopt}
\begin{split}
\inf_{\rho_{ABE}, \hat A_x, \hat B_y} & \sum_{i=1}^m \frac{\omega_i}{t_i \ln(2)} \sum_{a=0}^1 \inf_{\hat Z_{a,i}} \Tr\left[ \rho_{AE} \left( \hat A_1^{(a)} \otimes (\hat Z_{i,a}+\hat Z_{i,a}^\dag + (1-t_i)\hat Z_{i,a}^\dag \hat Z_{i,a}) + t_i(\mathbb{I}_A\otimes \hat Z_{i,a}\hat Z^\dag_{i,a})\right)\right]\\
\text{s.t.}&\Tr(\rho_{AB} \frac{\id+a\hat A_x}{2} \otimes \frac{\id+b\hat B_y}{2}) = \mathcal{P}(a,b|x,y)\\
& \hat A_x ^2=\id\\
& \hat B_y ^2=\id
\end{split}
\end{align}

As explained in~\cite{brown_device-independent_2021}, this formulation generalizes easily to the case of random pre-processing. Specifically, when Alice flips her output bit $A_1$ with probability $\mathrm{p}$, like in the DIQKD protocol considered in the main text, a slight modification is required to estimate $H(A_1|E)$: the operators $\hat{A}_1^{(a)}$ are to be replaced by
\begin{equation}
\left( \hat{A}_1^{(0)}, \hat{A}_1^{(1)} \right)  \rightarrow \left( (1-p)\hat{A}_1^{(0)} + p\hat{A}_1^{(1)} , p\hat{A}_1^{(0)} + (1-p) \hat{A}_1^{(1)} \right).
\end{equation}

\subsubsection{The full hierarchy}

The polynomial optimization \cref{eq:polopt} can be relaxed into an SDP hierarchy by considering a single global optimization over all non-commuting variables $\hat A_x$, $\hat B_y$, $\hat Z_{a,i}$ and choosing an NPA level $\ell$~\cite{Pironio10}. This lead Ref.~\cite{brown_device-independent_2021} to define the following SDP bound on $H(A_1|E)$ given the statistics $\mathcal{P}$, which we refer to as the \textit{full hierarchy}:
\begin{equation}
\label{eq:fullstat1}
       \begin{split}
       H(A_1|E)(\mathcal{P})  \geq H(A_1|E)_{\mathrm{full},\ell,m}(\mathcal{P}) := \underset{\{x_j\}_j}{\inf} & \sum_{i=1}^{m}  \tr(\Gamma C_i),\\
       \text{s.t.}&\forall k, \tr(\Gamma A_k) = b_k(\mathcal{P})\\
       &\Gamma = \sum_j F_j x_j\geq 0
       \end{split}
\end{equation}
where $\Gamma$ is the NPA moment matrix, $x_j$ are moments of the non-commuting variables, $A_k$ and $F_j$ are constant matrices, $b_k(\mathcal{P})$ is a function of the distribution $\mathcal{P}$, and $C_i$ is such that
\begin{equation}
\begin{split}
\label{eq:HAE_SDP}
            \Tr(\Gamma C_i) &= \sum_a \frac{\omega_i}{t_i \ln(2)} \left(   \braket{\hat{A}_1^{(a)} \hat{Z}_{i,a}} + \braket{\hat{A}_1^{(a)} \hat{Z}^\dagger_{i,a}} + (1-t_i) \braket{\hat{A}_1^{(a)} \hat{Z}^\dagger_{i,a} \hat{Z}_{i,a}} + t_i \braket{\hat{Z}_{i,a} \hat{Z}^\dagger_{i,a}} \right).
\end{split}
\end{equation}
The dual of \cref{eq:fullstat1} can be written as
\begin{equation}
\label{eq:fullstat1dual}
\begin{split}
\underset{M, \{\lambda_k\}_k}{\sup} & \sum_k \lambda_k b_k(\mathcal{P})\\
\text{s.t.}&\forall j, \tr(M F_j) + \sum_k\lambda_k\tr(A_k F_j) = \sum_i \tr(F_j C_i)\\
& M\geq 0.
\end{split}
\end{equation}

As discussed in~\cite{brown_device-independent_2021}, the full hierarchy converges to the actual conditional entropy $H(A_1|E)$ in the limit $\ell,m\to\infty$. Furthermore, the dual solution provides a vector $\lambda$ which describes a hyperplane that is tangent to the isosurface $H(A_1|E)= $ cst, i.e.~an $I$-score as discussed in more details in \cref{sec:Iscore}.

However, this method is costly to implement because the SDP matrix size grows quickly with $\ell$ and $m$. Namely, when considering two binary observables for Alice and Bob and an NPA level $\ell$, the matrix $\Gamma$ is generated by a first row containing $(5 + 4m)^\ell$ elements. For $m=10$ and $\ell=2$ for instance, the matrix is of size $2025\times 2025$ with several millions elements. In practice, SDPs typically become difficult to solve on a laptop computer when the matrix size approaches $1000 \times 1000$. Moreover, the runtime of SDP solvers grows quickly with the matrix size. Hence, the applicability of this method is limited to small values of $m$.

To illustrate the structure of the full SDP and compare it to the variants below, we consider a toy model. Namely, let us consider a simplified scenario where Alice and Bob have just one measurement operator $\hat A$ and $\hat B$, let's set $m=2$ and $\ell=1$ and gather Eve's operators $(\hat Z_{1,1},\hat Z_{1,2},\hat Z_{1,1}^\dag,\hat Z_{1,2}^\dag)$ into $\hat Z_1$ and $(\hat Z_{2,1},\hat Z_{2,2},\hat Z_{2,1}^\dag,\hat Z_{2,2}^\dag)$ into $\hat Z_2$ for illustration purposes. In this case, the SDP matrix $\Gamma$ is of the form
\begin{equation}
\Gamma_{\text{full}} = 
\begin{pmatrix}
\braket{\hat{\mathbb{I}}} & \braket{\hat{A}} & \braket{\hat{B}} & \braket{\hat{Z}_1} & \braket{\hat{Z}_2} \\
\braket{\hat{A}} & \braket{\hat{A}} & \braket{\hat{A}\hat{B}} & \braket{\hat{A}\hat{Z}_1} & \braket{\hat{A}\hat{Z}_2} \\
\braket{\hat{B}} & \braket{\hat{A}\hat{B}} & \braket{\hat{B}} & \braket{\hat{A}\hat{Z}_1} & \braket{\hat{A}\hat{Z}_2} \\
\braket{\hat{Z}_1^\dag} & \braket{\hat{Z}_1^\dag\hat{A}} & \braket{\hat{Z}_1^\dag\hat{B}} & \braket{\hat{Z}_1^\dag\hat{Z}_1} & \braket{\hat{Z}_1^\dag\hat{Z}_2} \\
\braket{\hat{Z}_2^\dag} & \braket{\hat{Z}_2^\dag\hat{A}} & \braket{\hat{Z}_2^\dag\hat{B}} & \braket{\hat{Z}_2^\dag\hat{Z}_1} & \braket{\hat{Z}_2^\dag\hat{Z}_2}
\end{pmatrix}.
\end{equation}

\subsubsection{The split hierarchy}

The difficulty of solving the full SDP above was already raised in Ref.~\cite{brown_device-independent_2021}, who proposed the following workaround. Noticing that it is possible to obtain a lower bound on $H(A_1|E)(\mathcal{P})$ by exchanging the minimum and the sum in Eq.~\eqref{eq:fullstat1} leads to an optimization hierarchy of the following form, which we refer to as the \textit{split hierarchy}:
\begin{equation}
\label{eq:fullstat2}
       \begin{split}
       H(A_1|E)(\mathcal{P})  \geq H(A_1|E)_{\textrm{split},\ell,m}(\mathcal{P}) :=  \sum_{i=1}^{m} \underset{\{x_j\}_j}{\inf}  & \Tr(\Gamma_i \tilde{C}_i)\\
       \text{s.t.}&\forall k, \tr(\Gamma_i A_k) = b_k(\mathcal{P})\\
       &\Gamma_i = \sum_j F^i_j x_j\geq 0,
       \end{split}
\end{equation}
where $\Gamma_i$ are sub-matrices extracted from the moment matrix $\Gamma$, and $\tilde{C}_i$ are constant matrices such that $\Tr(\Gamma C_i) = \Tr(\Gamma_i \tilde{C}_i)$. Here, the optimization is split into $m$ minimizations and only the variables $x_j$ appearing in the relevant constraints need to be taken into account in each minimization. In particular, no moments involving $\hat Z_{k,a}$ operators with $k\neq i$ is required.

Since \cref{eq:fullstat2} involves several SDP optimizations, it does not admit a single SDP dual. Nevertheless, a dual can be written for each $i$ as
\begin{equation}
\label{eq:fullstat2dual}
\begin{split}
\underset{M_i, \{\lambda^i_k\}_k}{\sup} & \sum_k \lambda^i_k b_k(\mathcal{P})\\
\text{s.t.}&\forall j, \tr(M_i F^i_j) + \sum_k\lambda^i_k\tr(A_k F_j) = \tr(F_j \tilde C_i)\\
& M_i\geq 0.
\end{split}
\end{equation}

The size of each moment matrix $\Gamma_i$ for this bound scales better with the hierarchy level $\ell$ and the parameter $m$. For instance, in the same case involving two binary measurements for Alice and Bob, each SDP matrix is of size $(5+4)^\ell \times (5+4)^\ell$ elements, i.e.~$81\times81$ when $m=10,\ell=2$. The computation of a single bound now involves solving $m$ distinct SDP problems. Each of these SDP problems can however by solved independently, which allows reaching large values of $m$ without a drastic increase in computation time.

However, this method comes with several drawbacks:
\begin{enumerate}
\item The convergence to the true value $H(A_1|E)$ is not guaranteed anymore, and it can be checked that the gap can be quite significant, see Secs.~\ref{sec:compareBounds} and \ref{seq:lilbenchmark} below.
\item Since there is no single SDP dual, the method does not provide a Bell expression that can serve as a certificate for the entropy lower bound $H(A_1|E)_{\text{split},\ell,m}(\mathcal{P})$. This is an issue for finite size analysis since this Bell expression is at the root of security proofs for DIQKD protocol, such as the one derived in this manuscript.
\item In practice, we observe that some of solvers may can have trouble converging to the solution of some SDPs in the split method. This reduces the precision on the conditional entropy $H(A_1|E)_{\text{split},\ell,m}(\mathcal{P})$ and can impact the security guarantee of DIQKD.
\end{enumerate}


In the simplified picture introduced above with $m=2$, the optimization is split in two independent optimizations, with respectively the operators $\{\hat{\mathbb{I}},\hat{A},\hat{B},\hat{Z}_1\}$ and $\{\hat{\mathbb{I}},\hat{A},\hat{B},\hat{Z}_2\}$. Computing the split bound then amounts to solving two different SDPs with the following matrices:

\begin{equation}
\begin{split}
\Gamma_{\text{split},1} = \begin{pmatrix}
\braket{\hat{\mathbb{I}}} & \braket{\hat{A_1}} & \braket{\hat{B_1}} & \braket{\hat{Z}_1} \\
\braket{\hat{A_1}} & \braket{\hat{A_1}} & \braket{\hat{A_1}\hat{B_1}} & \braket{\hat{A_1}\hat{Z}_1} \\
\braket{\hat{B_1}} & \braket{\hat{B_1}\hat{A_1}} & \braket{\hat{B_1}} & \braket{\hat{B_1}\hat{Z}_1} \\
\braket{\hat{Z}_1^\dag} & \braket{\hat{Z}_1^\dag\hat{A_1}} & \braket{\hat{Z}_1^\dag\hat{B_1}} & \braket{\hat{Z}_1^\dag\hat{Z}_1} \\
\end{pmatrix} \\[1em]
\Gamma_{\text{split},2} = \begin{pmatrix}
\braket{\hat{\mathbb{I}}} & \braket{\hat{A_2}} & \braket{\hat{B_2}} & \braket{\hat{Z}_2} \\
\braket{\hat{A_2}} & \braket{\hat{A_2}} & \braket{\hat{A_2}\hat{B_2}} & \braket{\hat{A_2}\hat{Z}_2} \\
\braket{\hat{B_2}} & \braket{\hat{B_2}\hat{A_2}} & \braket{\hat{B_2}} & \braket{\hat{B_2}\hat{Z}_2} \\
\braket{\hat{Z}_2^\dag} & \braket{\hat{Z}_2^\dag\hat{A_2}} & \braket{\hat{Z}_2^\dag\hat{B_2}} & \braket{\hat{Z}_2^\dag\hat{Z}_2} \\
\end{pmatrix}
\end{split}.
\end{equation}

\subsubsection{The block hierarchy}

Considering the polynomial optimization \cref{eq:polopt}, we observe that the optimization over the $\hat Z_{a,i}$ operators occurs within the sum over $i$, suggesting that these operators could be optimized independently of each other. Still, it is clear that their optimizations should remain consistent with a single global state $\rho_{ABE}$, and unique operators $\hat A_x$ and $\hat B_y$ for Alice and Bob. We thus propose to bound $H(A|E)$ with an SDP of the following form
\begin{equation}
\label{eq:fullstat3}
       \begin{split}
       H(A_1|E)(\mathcal{P})  \geq H(A_1|E)_{\textrm{block},\ell,m}(\mathcal{P}) :=  \underset{\{x_j\}_j}{\inf}  \sum_{i=1}^{m} & \Tr(\Gamma_i \tilde{C}_i)\\
       \text{s.t.}&\forall k, \tr(\Gamma_1 A_k) = b_k(\mathcal{P})\\
       &\forall i, \Gamma_i = \sum_j F^i_j x_j\geq 0.
       \end{split}
\end{equation}
The dual of this SDP can be expressed as
\begin{equation}
\label{eq:fullstat3dual}
\begin{split}
\underset{\{M_i\}_i, \{\lambda_k\}_k}{\sup} & \sum_k \lambda_k b_k(\mathcal{P})\\
\text{s.t.}&\forall i, j,  \tr(M_i F^i_j) + \delta_{i,1} \sum_k\lambda_k\tr(A_k F^i_j) = \sum_i \tr(F^i_j \tilde C_i)\\
&\forall i, M_i\geq 0.
\end{split}
\end{equation}

Effectively, this formulation combines the SDP matrices of the split method and solves them together in a single SDP optimization. In particular, it enforces that the moments involving no $Z$ operator appearing in several moment matrices $\Gamma_i$, such as $\langle \hat A_1 \hat A_2\rangle$, to take the same value in all matrices, as they should. At the same time, the size of the SDP blocks is unchanged, and thus remains independent of the $m$ parameter, which allows for an efficient resolution of the problem by an SDP solver. Indeed, the computational advantage of the split method with respect to the full one is preserved, as detailed in Sec.\ref{sec:compareBounds}.

By construction, we have the following inequalities:
\begin{equation}
H(A_1|E)_{\mathrm{split},\ell,m}\leq H(A_1|E)_{\mathrm{block},\ell,m}\leq H(A_1|E)_{\mathrm{full},\ell,m}.
\end{equation}
Moreover, we show in the next section that the block method actually converges to the true entropy $H(A_1|E)$, i.e.
\begin{equation}
\lim_{m\to\infty} \lim_{\ell\to\infty} H(A_1|E)_{\text{block},\ell,m} = H(A_1|E),
\end{equation}
so that its bounds are potentially as tight as the full hierarchy. Furthermore, the block SDP's dual provides a certificate for the entropy bound in the form of a single Bell inequality, which is crucial for the finite-size analysis of this manuscript as discussed earlier. Therefore, the block method combines the advantages of both the split and the full methods, providing an efficient way to bound the conditional von Neuman entropy in a tight fashion, and with a certificate.

Note that since the infimum over the $\hat Z_{i,a}$ operators in \cref{eq:polopt} lies inside the sum over both $i$ and $a$, one could also further split each $\Gamma_i$ matrix into two blocks, corresponding to each value of $a$. However, we noticed that this variant is not always beneficial in terms of computation time, because a higher level of the NPA hierarchy might be needed to obtain a good bound. Therefore, we choose here to work with $m$ blocks.

Coming back to the illustration above, the moment matrix associated to the block SDP can be put in the block-diagonal form
\begin{equation}
\Gamma_{\text{block}} = 
\begin{pmatrix}
\begin{pmatrix}
\braket{\hat{\mathbb{I}}} & \braket{\hat{A}} & \braket{\hat{B}} & \braket{\hat{Z}_1} \\
\braket{\hat{A}} & \braket{\hat{A}} & \braket{\hat{A}\hat{B}} & \braket{\hat{A}\hat{Z}_1} \\
\braket{\hat{B}} & \braket{\hat{B}\hat{A}} & \braket{\hat{B}} & \braket{\hat{B}\hat{Z}_1} \\
\braket{\hat{Z}_1^\dag} & \braket{\hat{Z}_1^\dag\hat{A}} & \braket{\hat{Z}_1^\dag\hat{B}} & \braket{\hat{Z}_1^\dag\hat{Z}_1} \\
\end{pmatrix} & 0 \\[1em]
0 & 
\begin{pmatrix}
\braket{\hat{\mathbb{I}}} & \braket{\hat{A}} & \braket{\hat{B}} & \braket{\hat{Z}_2} \\
\braket{\hat{A}} & \braket{\hat{A}} & \braket{\hat{A}\hat{B}} & \braket{\hat{A}\hat{Z}_2} \\
\braket{\hat{B}} & \braket{\hat{B}\hat{A}} & \braket{\hat{B}} & \braket{\hat{B}\hat{Z}_2} \\
\braket{\hat{Z}_2^\dag} & \braket{\hat{Z}_2^\dag\hat{A}} & \braket{\hat{Z}_2^\dag\hat{B}} & \braket{\hat{Z}_2^\dag\hat{Z}_2} \\
\end{pmatrix}
\end{pmatrix}.
\end{equation}
Like in the case of the split SDP, the size of the SDP matrices scales well with the hierarchy level $\ell$ and the parameter $m$, namely as $(5+4)^\ell\times(5+4)^\ell$ for two binary observables per party, which is independent of $m$. For $m=10$, $\ell=2$ the SDP thus involves blocks of size $81\times81$.

\subsection{Convergence of the Block SDP hierarchy}

\begin{proposition}
The block SDP hierarchy converges to the von Neumann entropy, i.e.
\begin{equation}
\lim_{m\to\infty} \lim_{\ell\to\infty}H(A_1|E)_{\text{block},\ell,m} = H(A_1|E).
\end{equation}
\end{proposition}
\begin{proof}
First, we notice that the objective function 
\begin{equation}
\sum_{i=1}^m \frac{\omega_i}{t_i \ln(2)} \sum_{a=0}^1 \Tr\left[ \rho_{AE} \left( \hat A_1^{(a)} \otimes (\hat Z_{i,a}+\hat Z_{i,a}^\dag + (1-t_i)\hat Z_{i,a}^\dag \hat Z_{i,a}) + t_i(\mathbb{I}_A\otimes \hat Z_{i,a}\hat Z^\dag_{i,a})\right)\right]
\end{equation}
only involves selected products of the operator variables $\hat A_x$, $\hat B_y$, $\hat Z_{a,i}$. For instance, no product of $\hat Z_{i,a}$ or $\hat Z_{i,a}^\dag$ operators with different values of $i$ or $a$ appears here. Our polynomial optimization problem \cref{eq:polopt} is thus an instance of a sparse noncommutative polynomial optimization~\cite{Klep22}. 

Second, we notice that the block SDP hierarchy \cref{eq:fullstat3} corresponds to the SDP relaxation of this sparse optimization for the families of operators $I_i=\{\hat A_x, \hat B_y, \hat Z_{i,a}\}_{a,x,y}$.

Third, we check that these families of operators satisfy the Running Intersection Property (RIP)~\cite{Lasserre06}
\begin{equation}
(I_{i-1} \cup I_i) \subset I_j, j<i
\end{equation}
for all $i=2,\ldots,m$.

Theorem 3.3 of \cite{Klep22} then guarantees that this SDP hierarchy converges to the noncommuting polynomial optimum when $\ell\to\infty$, which describes the Gauss-Radau quadrature of $H(A_1|E)$ at fixed $m$. Since the quadrature itself converges to $H(A_1|E)$ as $m\to\infty$, the whole method converges to the conditional von Neumann entropy $H(A_1|E)$.
\end{proof}

Note that the convergence of the block SDP \cref{eq:fullstat3} to the von Neumann entropy is a consequence of the structure of the polynomial optimization \cref{eq:polopt}. In particular, optimizations sharing a similar sparse structure can also be described by a converging block SDP hierarchy. This applies to the von Neumann entropy hierarchy decribed in~\cite{Kossmann25} as well as the hierarchy presented in~\cite{Hahn24} to bound Rényi entropies. The block SDP technique thereby offers a way to improve these methods significantly.

In addition to the block restriction, we notice that the numerical bounds we obtain in the cases considered in this manuscript converge to the expected entropy when the NPA level used for Eve's operators $\hat Z_{i,a}$ is restricted to 1. This further simplifies the SDP hierarchy, because only products of Alice and Bob's operators need to be considered when increasing the hierarchy level. We conjecture that the hierarchy indeed convergence under this additional restriction.

\subsection{Comparison between the different bounds}\label{sec:compareBounds}

In order to compare the full, split and block methods, we consider the problem of bounding the joint conditional von Neumann entropy $H(A_1,B_1|E)$ of Alice and Bob's first measurements' outcomes as a function of the CHSH score $S$. To our knowledge, the precise evaluation of this quantity has so far resisted an analytical resolution. Nevertheless, a numerical upper bound was given in~\cite{Bhavsar23}, conjectured to be tight. Earlier SDP bounds managed to recover this upper bound up to two digits of precision, but couldn't go further. Using the block hierarchy, we recover this upper bound from below up to numerical precision, thus confirming numerically the validity of the conjecture.

In order to bound the joint entropy $H(A_1B_1|E)$ rather than the entropy of Alice's outcomes only $H(A_1|E)$, a small modification of the optimization problem needs to be considered. Namely, four nonhermitian operators $\hat Z_{i,a,b}$ acting on Eve's Hilbert space need to be introduced for each value of $i$, one for each pair of outcomes $(a,b)$. The polynomial objective function then takes the form
\begin{align}\label{eq:poloptjoint}
\sum_{i=1}^m \frac{\omega_i}{t_i \ln(2)} \sum_{a,b=0}^1 \inf_{\hat Z_{i,a,b}} \Tr\left[ \rho_{ABE} \left( \hat A^{(a)} \otimes \hat B^{(b)}\otimes (\hat Z_{i,a,b}+\hat Z_{i,a,b}^\dag + (1-t_i)\hat Z_{i,a,b}^\dag \hat Z_{i,a,b}) + t_i(\mathbb{I}_A\otimes \mathbb{I}_B\otimes \hat Z_{i,a,b}\hat Z^\dag_{i,a,b})\right)\right].
\end{align}
The rest of the analysis is unchanged, as describe in Ref.~\cite{brown_device-independent_2021}.

In \cref{fig:plotHAB} we plot the lower bound on $H(A_1B_1|E)$ obtained when minimizing this objective function as a function of the CHSH score $S$ with the different methods. For both the block and split methods, we use the second hierarchy level for Alice and Bob, and the first one for Eve, i.e.~monomials in $\mathcal{O}_2 = [\hat{\mathbb{I}},\hat{A}_1,\hat{A}_2,\hat{A_1}\hat{A_2}, \hat{A_2}\hat{A_1}] \otimes [\hat{\mathbb{I}},\hat{B}_1,\hat{B}_2,\hat{B}_1 \hat{B}_2, \hat{B}_2 \hat{B}_1] \otimes [\hat{\mathbb{I}},\hat{Z}]$, where $\hat Z$ denotes all the $\hat Z$ operators that are relevant to the considered SDP matrix. Due to the higher computational cost of the full SDP method, we restrict the hierarchy level to level 1 for all parties, i.e.~choose monomials in $\mathcal{O}_1 = [\hat{\mathbb{I}},\hat{A}_1,\hat{A}_2] \otimes [\hat{\mathbb{I}},\hat{B}_1,\hat{B}_2] \otimes [\hat{\mathbb{I}},\hat{Z}]$. Despite this restriction, computation cost limits the full SDP computations to $m\leq 12$. For this reason, the bounds obtained with the full SDP are generally lower than the other ones obtained by the block method (and take longer to compute). While the full and split lower bounds show a finite gap compared to the upper bound, the bound computed with the block bound closes this bound for all CHSH scores.

\cref{fig:plotHABdetails} shows the details of the convergence of the block bound to the upper bound as $m$ increases for the CHSH value $S=2.3$. A comparison of the computation time for the different methods is also provided, showing important computational cost of the full method. Surprisingly, the block method is the fastest of the three. We propose another comparison between the three SDP hierarchies in the context of key rate optimization in \cref{seq:lilbenchmark}.

\begin{figure}
\includegraphics[width=0.6\linewidth]{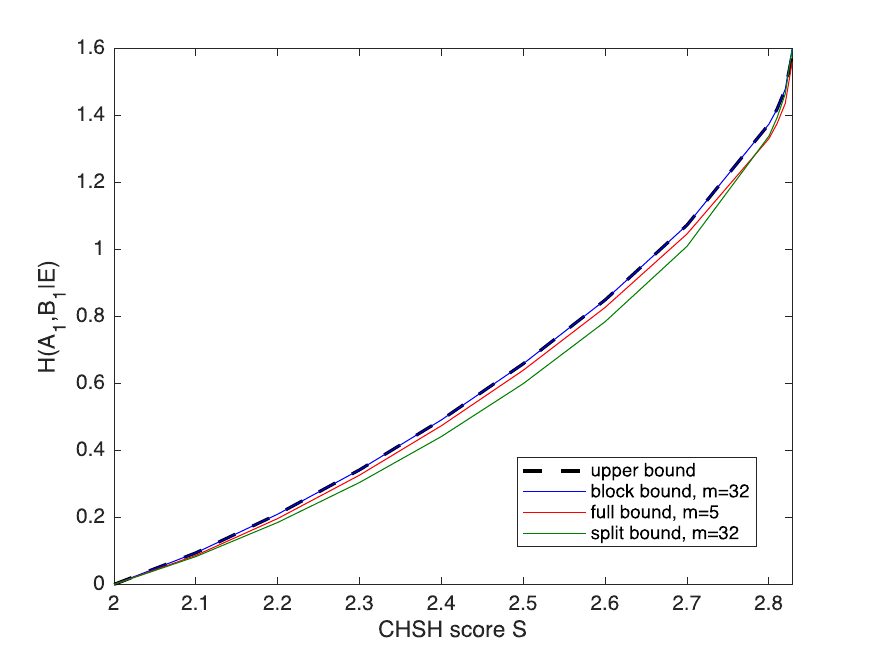}
\caption{Bounds on the joint von Neumann entropy of Alice and Bob's outcomes as a function of the CHSH score. Here, the computation time is limited to $\sim$3 minutes per point. Both the full and split bounds are suboptimal. The lower bound obtained with the block method matches the numerical upper bound from Ref.~\cite{Bhavsar23} up to numerical precision, as shown in \cref{fig:plotHABdetails}.}
\label{fig:plotHAB}
\end{figure}

\begin{figure}
\includesvg[width=\linewidth]{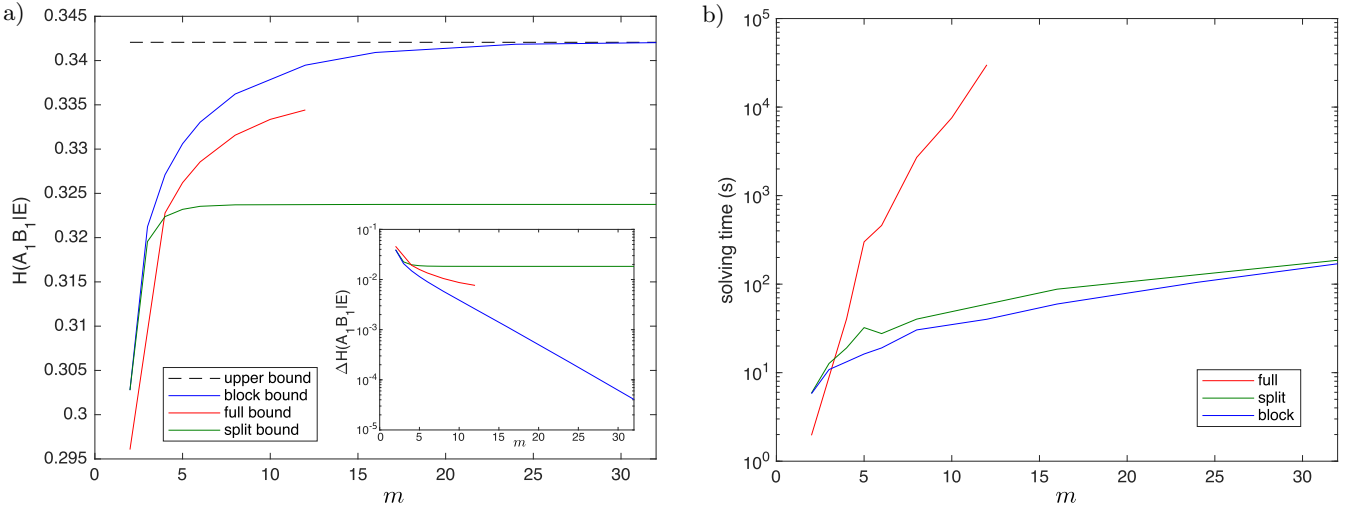}
\caption{a) Evolution of the lower bounds for the different SDP approaches as $m$ increases for $S=2.3$. The inset shows the difference between the upper bound and the various lower bounds. The split bound converges to a suboptimal value. The full method quickly becomes intractable, taking more than 10 hours to compute one point at $m=12$. The bounds obtained with the full method are smaller than the other ones from the other methods here because of the restricted hierarchy level $\mathcal{O}_1$: using level $\mathcal{O}_2$ allows the full method to recover the bounds obtained with the block method, although for an even higher computational cost, thus restricting $m$ even further. The block bound converges exponentially fast to the conjectured value up to numerical precision ($\sim 10^{-5}$). b) Computation time for each method as a function of $m$. The block method requires a single call to the SDP solver, thus limiting overhead and solving generally faster than the split method.}
\label{fig:plotHABdetails}
\end{figure}

\section{$I$-score} \label{sec:Iscore}
As discussed in \cref{app:block}, the block SDP hierarchy provides a bound on the conditional von Neumann entropy $H(A_1|E)$ as a function of the behavior $\mathcal{P}$. By strong dualty, its dual solution can be interpreted as a Bell expression providing a tight certificate of this entropy. Here, we discuss this Bell expression, which we refer to as the $I$-score. In particular, we illustrate geometrically why this $I$-score is in general better than the CHSH expression to bound the conditional von Neumann entropy. We then describe an iterative method that takes advantage of the $I$-score obtained by the block bound to optimize $H(A_1|E)(\mathrm{x})$ and the asymptotic key rate as a function of the setup and pre-processing parameters $\mathrm{x}=(g,\alpha_1,\alpha_2,\beta_0,\beta_1,\beta_2)$ and $\mathrm{p}$ while reducing the number of SDP resolutions. Finally, we benchmark this method on the three SDP hierarchies introduced above.

\subsection{Optimized entropy bound through the $I$-score}
\label{sec:I_score}

When computing $H(A_1|E)_{\text{block},\ell,m}$ for a behavior $\mathcal{P}_0$, the solution of the dual \cref{eq:fullstat3dual} provides a vector $\lambda_k$ from which we can define the Bell expression
\begin{equation}\label{eq:Iscore}
I(\mathcal{P})=\sum_k \lambda_k b_k(\mathcal{P}) = \innerproduct{\lambda}{\mathcal{P}},
\end{equation}
for all behaviors $\mathcal{P}$, where $\innerproduct{.}{.}$ is the Euclidean scalar product in the space of behaviors. As a Bell expression, the $I$-score admits two \textit{local bounds} $I_{C1}$,$I_{C2}$ and two \textit{quantum bounds} $I_{Q1}$,$I_{Q2}$ corresponding to the value of the $I$-score of the edges of the quantum set, see \cref{fig:bound_III}a).

For every value $I$, the $I$-score defines a hyperplane $D(I)=\{\mathcal{P}, \innerproduct{\lambda}{\mathcal{P}} = I\}$ of constant $I$-score in probability space. Strong duality implies that $H(A_1|E)_{\text{block},\ell,m}(\mathcal{P})\geq \innerproduct{\lambda}{\mathcal{P}_0}$ for all $\mathcal{P}\in D(I_0)$ with $I_0=I(\mathcal{P}_0)$. This ensures that the hyperplane $D(I_0)$ is tangent to the isosurface $H(A_1|E)_{\text{block},\ell,m}(\mathcal{P}) = I_0$ at $\mathcal{P}=\mathcal{P}_0$. More generally, the entropy as a function of the $I$-score can be defined as the minimum of $H(A_1|E)$ on the hyperplane $D(I)$ written as
\begin{equation}
    H(A_1|E)(I) = \min_{\mathcal{P} \in D(I) } H(A_1|E)(\mathcal{P}).
\end{equation}
A typical sketch of this function is shown in \cref{fig:bound_III}b).

\begin{figure}
  \centering
  \includesvg[width=0.95\linewidth]{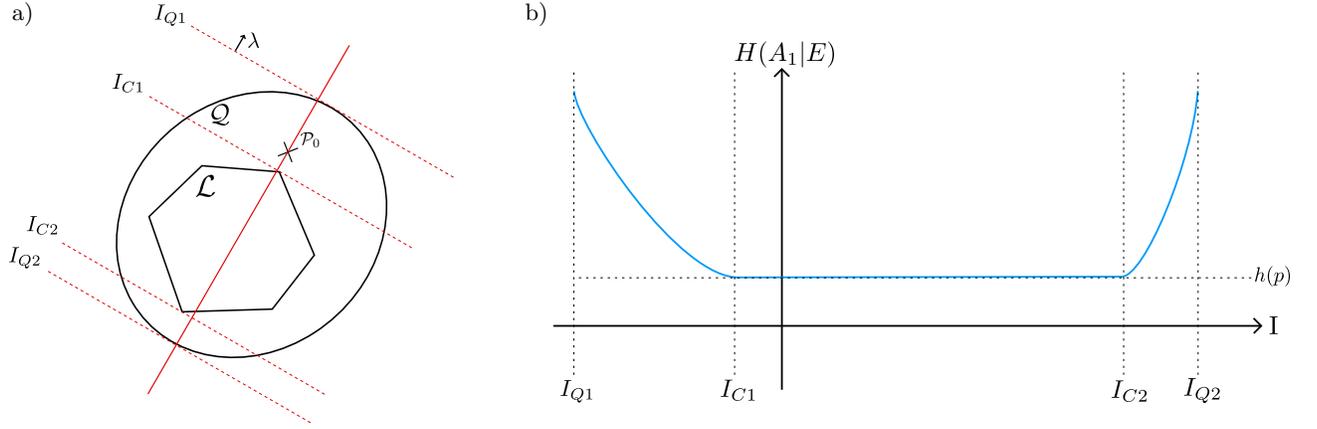}
  \caption{a) Sketch of the functional from of the $I$-score $I(\mathcal{P})$ in probability space. We observe the different bounds along with their relationship with the vector $\lambda$ and the local and quantum sets. b) Sketch of function $H(A_1|E)(I)$ in presence of noisy pre-processing, with the different bounds illustrated in a). We observe that in the local set, the function has the value $h(\mathrm{p})$ corresponding to the entropy generated solely by noisy-preprocessing.}
  \label{fig:bound_III}
\end{figure}

Let us illustrate why the CHSH expression is not always the best suited to estimate $H(A_1|E)$ for the computation of the asymptotic key rate given a behavior $\mathcal{P}_0$. We remind that the asymptotic key rate is given by $r = H(A_1|E) - H(A_1|B_0)$, and that $H(A_1|B_0)$ is fully determined by the statistics $\mathcal{P}_0$; hence improvements are expected from a better estimation of $H(A_1|E)$. In Fig.\ref{fig:CHSH_bad_I}, we can see that the local CHSH hyperplane (in red) is not necessarily colinear to the tangent of the iso-surface of $H(A_1|E)(\mathcal{P}_0)$. This implies that the minimisation of the entropy as a function of the CHSH score $H(A_1|E)(S)$ must take into accounts statistics with lower value of $H(A_1|E)$ than the minimization of $H(A_1|E)(I)$, therefore resulting in a suboptimal estimation of the von Neumann entropy for this point. When the full probability distribution $\mathcal{P}_0$ is accessible, it is thus more appropriate to use the $I$-score computed from $\mathcal{P}_0$ through \cref{eq:Iscore} rather than the CHSH score to estimate $H(A_1|E)$.

\begin{figure}
  \centering
  \includesvg[width=0.3\linewidth]{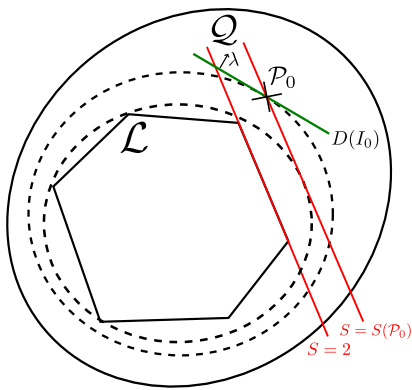}
  \caption{A slice of the probability distribution set. $\mathcal{L}$ is the local set, $\mathcal{Q}$ is the quantum set. The dotted lines correspond to two iso-surfaces $\{ \mathcal{P}, H(A_1|E)(\mathcal{P}) = \text{cst}  \}$, one inside the other, and $\mathcal{P}_0$ is a point in the distribution set. The green line corresponds to the hyperplane tangent to the isosurface at the point $\mathcal{P}_0$. The two red lines correspond to two hyperplanes containing all the distributions with the same CHSH score. We observe that for the CHSH score associated with the statistics $\mathcal{P}_0$, the hyperplane touches the surface with lower $H(A_1|E)$ (the one inside the other). The green hyperplane, however, excludes a larger conditional entropy iso-surface.}
  \label{fig:CHSH_bad_I}
\end{figure}

Now, we describe a first-order optimization scheme utilizing the $I$-score to optimize $H(A|E)(\mathcal{P}(x))$ with respect to the parameter of the setup $x$.

\subsection{Optimization of the key rate over the circuit parameters using the SDP dual vector}
\label{seq:FirstOrderOpti}

Knowing how to estimate $H(A_1|E)$ for a probability distribution $\mathcal{P}_0$, we are now interested in finding the parameters of the physical setup $\mathrm{x}$, such that the corresponding probability distribution $\mathcal{P}(\mathrm{x})$ maximizes the lower bound on $H(A_1|E)(\mathcal{P}(\mathrm{x}))$ provided by one of the methods described in \cref{app:block}. Computing these lower bounds requires solving an SDP, which can be computationally expensive. Therefore, when using a general optimization algorithm that calls this quantity frequently, the overall cost can become intractable. Thus, following~\cite{brownPrivate}, we develop a first-order method to optimize it with respect to $\mathrm{x}$ with the help of the previously defined $I$-score.

Starting with an initial parameter $\mathrm{x}_0$, $H(A_1|E)(\mathcal{P}(\mathrm{x}_0))$ is estimated using an SDP whose dual solution defines the $I$-score $I_0(\mathcal{P}(\mathrm{x}))=\innerproduct{\lambda}{\mathcal{P}(\mathrm{x})}$, c.f.~\cref{eq:fullstat1dual,eq:fullstat3dual} in the full and block cases. When using the split method, we use here $\lambda_k=\sum_{i=1}^M \lambda_k^i$, knowing that this $I$-score may not bound $H(A_1|E)$ optimally.

Considering small variations in parameters $\mathrm{x}$, we can estimate the induced change in the estimation of $H(A_1|E)$ to first order through the variation of the parameter $I$ by computing
\begin{equation}
    \alpha = \frac{d H(A_1|E)}{dI}(I_0(\mathrm{x}_0)).
\end{equation}
This characterizes how $H(A_1|E)(I_0)$ behaves in the vicinity of $I_0(\mathrm{x}_0)$. Importantly, the key rate does not solely depend on $H(A_1|E)$, but also on $H(A_1|B_0)$. Therefore, we consider the following function
\begin{equation}
    \phi(\mathrm{x}) = \alpha I(\mathcal{P}(\mathrm{x})) - H(A_1|B_0)(\mathcal{P}(\mathrm{x})) = \alpha \innerproduct{\lambda}{\mathcal{P}(\mathrm{x})} - H(A_1|B_0)(\mathcal{P}(\mathrm{x})),
\end{equation}
which captures locally the dependence of the asymptotic key rate on the system parameters $\mathrm{x}$, and can be defined with few calls to an SDP solver (to estimate $I(\mathcal{P}(\mathrm{x}))$ and $\alpha$). Optimizing this function, using \textit{Nelder-Mead} or \textit{Gradient descent}, is really fast and efficient. It gives us a new parameter $\mathrm{x}_1$ likely to provide a higher key rate.

Note that we have to be careful here when using noisy pre-processing, because the SDP estimation of $H(A_1|E)$ depends on the parameter $\mathrm{p}$. So, we fix this parameter during this optimization. The noisy pre-processing parameter is then optimized independently of the other parameters. In practice, we implement a grid refinement search.

Overall, we thus optimize the parameters $\mathrm{x}$ and $\mathrm{p}$ with the following strategy: start with $(\mathrm{x}_0,\mathrm{p}_0)$, estimate $H(A_1|E)$ to obtain the $I_0$-score, then run the first-order optimization to get $(\mathrm{x}_1,\mathrm{p}_0)$, then optimize the noisy pre-processing parameter to get $\mathrm{p}_1$ (because it is only one parameter, it is not costly to optimize $H(A_1|E)$ over $\mathrm{p}$ by SDP). With the new set of parameters $(\mathrm{x}_1,\mathrm{p}_1)$, we can restart the process a few times. In practice, we typically perform three such iterations.

\subsection{Benchmarking of the keyrate optimization with the different SDP methods}
\label{seq:lilbenchmark}

Here, we benchmark the first-order keyrate optimization with the three methods described in \cref{app:block}. We implement these and any subsequent methods in the Julia \cite{bezanson2017julia} programming language, using JuMP \cite{Lubin2023} as an interface for optimization and Mosek \cite{mosek} as a solver for the SDP. We rely on the numerical implementation of the NPA hierarchy written by Erik Woodhead's~\cite{Woodhead} and use the Julia package FastGaussQuadrature.jl \cite{GaussRadau} for the computation of Gauss-Radau quadratures. In practice, when declaring our optimization problem in JuMP, we define it in its dual form, since we notice that it is dealt with more efficiently by the solver. We also note that, as mentioned in Ref.~\cite{brown_device-independent_2021}, it is sufficient to consider real moment matrices when running our computations. This is because for any feasible complex moment matrix $\Gamma$, the matrix ($\Gamma$ + $\Gamma^{\dagger}$)/2 is also a feasible moment matrix, that is real and with the same objective value. Our code is available at~\cite{gitlabCode}.

Concretely, we consider the setup studied in the main text, generating a two-mode squeezed vacuum that is then locally displaced and measured using NPNR measurements of efficiency $\eta$. For the purpose of this comparison, we place ourselves at efficiency $\eta = 0.85$ and optimize the estimation of $H(A_1|E)(\mathcal{P}(x))-H(A_1|B_0)$ using the 3 different methods with the approach described in Sec.\ref{seq:FirstOrderOpti}.

Furthermore, we consider the NPA level defined by the following operators $\mathcal{O}_2 = [\hat{\mathbb{I}},\hat{A}_0,\hat{A}_1,\hat{A_0}\hat{A_1}, \hat{A_1}\hat{A_0}] \otimes [\hat{\mathbb{I}},\hat{B}_0,\hat{B}_1,\hat{B}_0 \hat{B}_1, \hat{B}_1 \hat{B}_0] \otimes [\hat{\mathbb{I}},\hat{Z}]$, where $\hat Z$ denotes all the $\hat Z$ operators that are relevant to the considered SDP hierarchy. In the case of the full bound, the time and memory required for the computation limits us to values of $M\leq10$ when solving an SDP and to $M\leq 7$ when optimizing over the settings.

In \cref{fig:higher_hierarchy}a), we observe that the full and block bounds provide a similar result, while the split bound gives a significantly lower value. In particular, the split method converges to a suboptimal entropy bound and key rate.

We emphasize that the curves in this figure are the result of the non linear optimization defined in \cref{seq:FirstOrderOpti} applied just 3 times with a different SDP. Therefore, the final statistical point corresponding to each hierarchy is the result of a distinct optimization and corresponds to different setup parameters $\mathrm{x}$. This is the reason why the block bound may sometimes be higher than the full bound even if both are computed at the same hierarchy level here.

\begin{figure}
  \centering
  \includesvg[width=\linewidth]{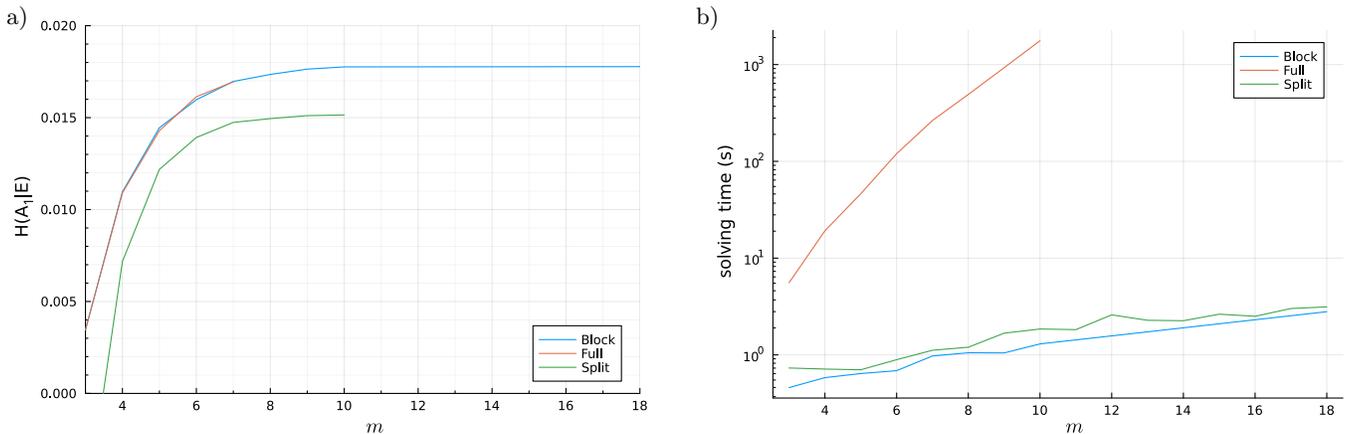}
  \caption{a) Comparison of the numerical values obtained after parameter optimization for the different bounds as a function of $m$. b) Benchmarking of the computation time for one evaluation of $H(A_1|E)$ using each SDP method as a function of $m$.}
  \label{fig:higher_hierarchy}
\end{figure}

In \cref{fig:higher_hierarchy}b), we also compare the computation time needed by each SDP hierarchy. We see that the full bound takes order of magnitude more time to compute than the split and block bounds. A similar scaling is observed for the split and block bounds, with a slight advantage for the block method that can be attributed to the reduce amount of overhead achieved by calling the SDP solver once instead of $m$ times.

These results are consistent with the ones reported in \cref{sec:compareBounds}, confirming that the block method combines the precision of the full bound and the efficiency of the split bound. In fact, the improved efficiency allows us to obtain bounds that are better than those accessible by our computation capabilities with the full method. Moreover, as stated earlier, a single tight Bell inequality can be extracted from the block bound, which is not possible with the split bound. For all these reasons, we use the block bound in this work to estimate conditional von Neumann entropies.

\section{Finite size full statistics}
\label{seq:finite_fullstat}

In this section, we present a security proof for the DIQKD protocol described in the main text that takes into account not solely the CHSH score of a setup but its complete behavior $\mathcal{P}$. Our analysis builds on the finite size security proof developed in Ref.~\cite{nadlinger_device-independent_2022}, summarized in its Proposition 1. This result relies on the EAT to bound the conditional smooth min-entropy $H^{\epsilon}_{\min}(A^n|X^nY^nE)$ as a function of the CHSH score, and thus provide a secure key length $l$ valid in presence of coherent attacks. In order to take into account the complete behavior $\mathcal{P}$, we first define relevant events, followed by EAT channels and min-tradeoff functions on which to apply the EAT to take advantage of the full statistics. This allows us to obtain a secret key length, as described in \cref{sec:keylength}. Finally, we discuss in more details the computation of several terms appearing in the key length formula.

\subsection{EAT ingredients with full statistics}

A central component of the DIQKD proof in Ref.~\cite{nadlinger_device-independent_2022} is the register $U=0,1,\perp$, which encodes the Bell score of each round during the protocol. The fact that this parameter only takes three values is possible because the CHSH Bell expression can be cast as a binary game. In order to access the complete observed probabilities $\mathcal{P}$, we need to allow for more freedom at the level of parameter estimation. Let us thus define the event $k \in \mathcal{X}$, stored in the register $K$, with $\mathcal{X}$ the set of possible events
\begin{equation}
    \mathcal{X} = \{(a,b,x,y) \text{ with } a,b \in \{0,1\}, x,y \in \{1,2\}  \} \cup \{ \perp\}.
\end{equation}
This allows us to now precise the form of EAT channels in presence of full statistics.

\subsubsection{Test rounds channels and key generation rounds channels}

Starting with the test rounds, we define the test quantum channels $\mathcal{M}^{test} : R \rightarrow RABXYTK $ with fixed $T=1$ and $k = (a,b,x,y)$ the choice of inputs and the outputs. This channel takes as input a quantum state in register $R$ and outputs another quantum state in $R$, the inputs in registers $X,Y$, outputs in registers $A,B$, the nature of the round in register $T$, and the event for EAT in register $K$. Each test channel corresponds to one test round of the protocol, taking a quantum state as input and generating an event $(a,b,x,y)$. We denote by $q(a,b,x,y)$ the joint probability distribution where the inputs $x,y$ are the chosen settings and the outputs $a,b$ are the corresponding measurement results. The set of distributions generated by test channels $\mathcal{Q}^{test}$ is defined, in Ref.~\cite{Liu2021} as

\begin{equation}
    \mathcal{Q}^{test} = \{ q ,\text{  } \exists \mathcal{M}^{test}, \hat{\rho} \in \mathcal{L}(\mathcal{H}) \text{  s.t.\,} \left( \mathcal{M}^{test}(\hat{\rho})  \right)_K = \sum_{a,b,x,y} q(a,b,x,y) \ket{a,b,x,y} \bra{a,b,x,y}  \}.
\end{equation}
Here, the distribution $q \in \mathcal{Q}^{test}$ corresponds to a distribution of input/outputs $(a,b,x,y)$ compatible with quantum mechanics during a test round. $\left( \mathcal{M}^{test}(\hat{\rho})  \right)_K$ corresponds to the register $K$ after the application of the testing channel to $\hat{\rho}$. Because we define these quantum channels of test rounds as being uniform on the setting $x,y$, it implies that $\sum_{a,b} q(a,b,x,y) = q(x,y) = \frac{1}{4}$. Note that the symbol $K=\perp$ is not used in a test round, so we have $q(\perp)=0$.

To any distribution $q(a,b,x,y) \in \mathcal{Q}^{test}$ of inputs/outputs, we can associate a distribution $\mathcal{P}_q(a,b|x,y) \in \mathcal{Q}$ of outputs conditioned on inputs, a distribution in the black-box paradigm, using the relation:

\begin{equation}
    \mathcal{P}_q(a,b|x,y) = \frac{q(a,b,x,y)}{q(x,y)} = \frac{1}{4} q(a,b,x,y).
\end{equation}
This enables us to define important quantities as a function of the inputs/outputs $(a,b,x,y)$ generated by test quantum channel $\mathcal{M}^{test}$ such as

\begin{equation}
    \begin{split}
        H(A|E)(q) &= H(A|E)(\mathcal{P}_q)\\
        I(q) &= I(\mathcal{P}_q)
    \end{split}
\end{equation}

We define the \textit{key generation} event as $\perp$ and key generation channels $\mathcal{M}_{keygen} : R \rightarrow RABXYTK $ with fixed $T=0$, $K = \perp$, $X=1$, $Y=0$, a key generation channel.

\subsubsection{The infrequent sampling channel}

Each round of the protocol is either with probability $\gamma$ a test round, or with probability $1-\gamma$ a key generation round. Thus, the channels we consider are \textit{infrequent sampling channels} : $\mathcal{M} : R \rightarrow RABXYTK$ defined as

\begin{equation}
    \mathcal{M}(.) =  \gamma \mathcal{M}^{test}(.) + (1-\gamma) \mathcal{M}^{keygen}(.),
\end{equation}
with $\mathcal{M}^{test}$ a testing channel and $\mathcal{M}^{keygen}$ a key generation channel. These infrequent sampling channels generate distributions in the following set

\begin{equation}
    \mathcal{Q}^{\gamma} = \{ p ,\text{  } \exists \mathcal{M}, \hat{\rho} \in \mathcal{L}(\mathcal{H}) \text{  s.t.\,} \left( \mathcal{M}(\hat{\rho}) \right)_K = \sum_{k \in \mathcal{X}} p(k) \ket{k} \bra{k}  \}.
\end{equation}
On the one hand, with a probability distribution $q \in \mathcal{Q}^{test}$ from test round, we can generate $p \in \mathcal{Q}^{\gamma}$ for both rounds with the following transformation

\begin{equation}
\label{eq:definition_p}
     p(k) = \begin{cases} 
     \gamma q(k) & \text{if } k = (a,b,x,y), \\
    1-\gamma & \text{if } k = \perp.
       \end{cases}
\end{equation}
On the other hand, with a fixed probability $\gamma$, and a probability distribution $p \in \mathcal{Q}^\gamma_{\mathcal{G}}$, we can uniquely extract a probability distribution $q \in \mathcal{Q}^{test}$ of test round. It means that we can define the $I$-score and $H(A|E)(p)$ for any distribution $p$ from the definition given for $q$.

\subsubsection{Min-tradeoff functions}

To apply EAT, it is required to have access to a min-tradeoff function for our probability distribution set $\mathcal{Q}^\gamma_{\mathcal{G}}$. Taking $p \in \mathcal{Q}^\gamma_{\mathcal{G}}$, we extract the probability distribution $q \in \mathcal{Q}^{test} $ and define a min-tradeoff function $g$, in the sense of Lemma 3 from \cite{Liu2021}, as:

\begin{equation}
    g(q) \leq \inf_{(\hat{\rho}_{RE},\mathcal{M})} \{ H(AB|E)_{\mathcal{M}(\hat{\rho}_{RE})}: \left( \mathcal{M}^{test} (\hat{\rho}_{RE}) \right)_K = \sum_{a,b,x,y} q(a,b,x,y) \ket{a,b,x,y} \bra{a,b,x,y} \},
    \label{eq:mintradeoffLiu}
\end{equation}
with $\hat{\rho}_{RE}$ the quantum state input in the quantum channel $\mathcal{M}$. The min-tradeoff function $g$ provides a lower bound on the conditional entropy $H(AB|E)$ for all infrequent sampling channels and quantum states whose associated test channel yields the distribution $q$. Lemma 3 from  \cite{Liu2021} tells us that if we have such a min-tradeoff function $g$, then the following function $f$ is a min-tradeoff function for the infrequent sampling channels:

\begin{equation}
\label{eq:mintradeoff_full}
     f(\delta_k) = \begin{cases} 
          \frac{1}{\gamma} g(\delta_k) + (1-\frac{1}{\gamma})c_{\perp} & \text{if }k = (a,b,x,y) \\
          c_{\perp} & \text{if } k = \perp
       \end{cases},
\end{equation}
with $c_{\perp} \in \mathbb{R}$.

Let us now describe the min-tradeoff function that we use in the modified protocol to account for the full statistics.

\paragraph{Building the min-tradeoff function}

\cref{eq:mintradeoffLiu} requires a bound on $H(AB|E)$. We first use the following bound:
\begin{equation}
    H(AB|E) \geq H(A|E),
\end{equation}
so that we now look at the infinimum on $H(A|E)$ on an infrequent sampling channel that generates the distribution $q$ during test rounds. Note that in the infrequent sampling channel, the probability of Alice choosing the setting $x$ is different from that in testing rounds. We decide to bound this quantity using the $I$-score of the distribution $q$. So we need to consider the following quantity
\begin{equation}
\label{eq:infrequent_conditional_entropy}
    H(A|E)(I(q)) = (1-\gamma)  H_p(A_1|E)(I(q)) + \frac{\gamma}{2}  ( H(A_1|E)(I(q)) + H(A_2|E)(I(q))).
\end{equation}
The first term on the right-hand side denotes the minimum of the conditional entropy with noisy pre-processing over probability distribution $q$ with and $I$-score $I(q)$, corresponding to the key generation round coming with probability $1-\gamma$, in which Alice chooses the setting $x=1$. The second term in the right-hand side is the minimum of the entropy of test rounds coming with probability $\gamma$ over probability distribution $q$ with $I$-score $I(q)$, in which Alice chooses with probability 1/2 either the setting $x=1$ or $x=2$ and in absence of pre-processing.

To ensure that this quantity is indeed a min trade-off function, we recall the following relation
\begin{equation}
    H(A|E)(q) \geq H(A|E)(I(q)),
\end{equation}
so it means that we have the following bound:
\begin{equation}
    H(AB|E)(q) \geq H(A|E)(I(q)).
\end{equation}
In order to obtain an affine min-tradeoff function, we define the min trade-off function for $q \in \mathcal{Q}^{test} $ and for an $I$-score value $v \in  \mathbb{R}$ as
\begin{equation}
\label{eq:mintradeoff_testround}
    g_v(q) = H(A|E)(v) + (I(q)-v) \frac{dH(A|E)(I)}{dI}(v).
\end{equation}

The function $g_v$ corresponds to the tangent of $H(A|E)(I)$ at point $v$, and the derivative of the function is estimated numerically using finite differences as detailed below in~\cref{sec:fi_and_gi}. The transformation \cref{eq:mintradeoff_full} defines the min-tradeoff function $f_v$ on no-signaling statistics. It is important to note, however, that this does not identify the min-tradeoff function completely, because $f_v$ acts on all probability distributions, including signaling ones.

\paragraph{Parametrizing the signaling component of the min-tradeoff function}
In the CHSH scenario, the value of a Bell expression, such as the $I$-score, is generally parametrized by 8 components, allowing its evaluating on all no-signaling statistics. Collins-Gisin probabilites constitute one such possible parametrization~\cite{Collins04}. The space of full probability distributions $\mathcal{P}(a,b|x,y)$ is however 16-dimensional, which leads to an additional freedom in the description within this space. This generic phenomenon also applies to Bell games such as the min-tradeoff functions $g_v(q)$ and $f_v(q)$~\cite{Rosset14}.

For this reason, a complete description of the min-tradeoff function is given by constructing it through \cref{eq:mintradeoff_full,eq:mintradeoff_testround} on the fully parametrized $I$-score
\begin{equation}
\bar I(\mathcal{P}) = I(\mathcal{P}) + \sum_j \mathrm{f}_j \innerproduct{V_j}{\mathcal{P}},
\end{equation}
where the 8 linearly independent vectors $V_j \in \mathbb{R}^{16}$ are such that $\innerproduct{V_j}{\mathcal{P}} = 0$ on any no-signaling distribution $\mathcal{P}$ for all $j$. These vectors encode the nosignaling constraints $\sum_a q(a,b|x,1)=\sum_a q(a,b|x,2)$ and $\sum_b q(a,b|1,y)=\sum_b q(a,b|2,y)$. We refer to the corresponding min-tradeoff functions, which depend on the parameters $\vec{\mathrm{f}}=(\mathrm{f}_1,\ldots,\mathrm{f}_8)$ as respectively $\bar g_v$ and $\bar f_v$. We emphasize that by construction, $\bar I(\mathcal{P})=I(\mathcal{P})$ for all no-signaling behavior $\mathcal{P}$, $\bar g_v(p)=g_v(p)$ for all $p\in\mathcal{Q}^{test}$, and $\bar f_v(q)=f_v(q)$ for all $q\in\mathcal{Q}^\gamma$.

Now that we have detailed the events considered (taking into account all statistics and not just the win/loose outcomes of the CHSH game) and defined a min-tradeoff function based on the $I$-score, we are ready to modify the DIQKD security proof from \cite{nadlinger_device-independent_2022} to account for full statistics.

\subsection{Finite key length formula with full statistics}\label{sec:keylength}

In order to adapt the security proof derived in Proposition 2 of \cite{nadlinger_device-independent_2022}, we apply the general Theorem 3 of \cite{Liu2021} on our specific scenario involving the $I$-score using the events and min-tradeoff functions defined above (in place of Theorem 1 in Ref.~\cite{nadlinger_device-independent_2022} involving the CHSH score). This yields the following bound on the extractable number of secure key bits $l$:

\begin{equation}
\begin{split}
l &=   n g_v(I_{thr}) + n \inf\limits_{p \in \mathcal{Q}^{\gamma}}(\Delta(f_v,p) - (\alpha'-1)V(\bar f_v,p) )  \\ 
 & -n(\alpha'-1)^2 K_{\alpha'}(f_i) -n\gamma -n(\alpha"-1) \log^2(5) \\
& -\frac{1}{\alpha'-1}(\vartheta_{\epsilon_s'}+\alpha' \log(\frac{1}{\epsilon_{EA}})) - \frac{1}{\alpha"-1}(\vartheta_{\epsilon_s"} + \alpha" \log(\frac{1}{\epsilon_{EA}}) ) \\
 &-3\vartheta_{\epsilon_s-\epsilon_s'-2\epsilon_s"} - 5 \log(\frac{1}{\epsilon_{PA}})-EC-264
 \label{eq:finitesize}
\end{split}
\end{equation}
with 
\begin{equation}
    \begin{split}
        &\Delta(f_v,p) = H(A|E) ( I(p) ) - f_v(I(p)) \\
        & V(\bar f_v,p) = \frac{\ln(2)}{2} \left( \log(33) + \sqrt{2 + \text{Var}_p(\bar f_v)} \right)^2 \\
        & \text{Var}_p(\bar f_v) = \sum_x p(x)  (\bar f_v (x)-\mathbb{E}_p [ f_v ])^2 \\
        &K_{\alpha'}(f_v) = \frac{1}{6(2-\alpha')^3 \ln(2))} 2^{(\alpha'-1)(2+\max_{\mathcal{D}_{X_\gamma}}(f_v) - \min_{\mathcal{Q}^{\gamma}}(f_v))} \ln^3( 2^{(2+\max_{\mathcal{D}_{X_\gamma}}(f_v) - \min_{\mathcal{Q}^{\gamma}}(f_v))} +e^2)  \\
        &I(p) = \innerproduct{C}{\mathcal{P}(p)}\\
        &\vartheta_\epsilon = \log(\frac{1}{1-\sqrt{1-\epsilon^2}})
    \end{split}
\end{equation}
and $EC=|\mathrm{M}|$ is the length of the error correction syndrome.

In order to compute the keyrate for given behavior $\mathcal{P}$, a number of repetitions $n$ and security parameters $\epsilon_{s}, \epsilon'_s, \epsilon"_s, \epsilon_{EA}$ and $\epsilon_{PA}$, we proceed as follows. First, we compute the $I$-score by the block method using \cref{eq:Iscore}. We then fix the signaling components $\vec{\bar f}$ as described in \cref{sec:sigComp}. This allows us to fix $I_{thr}$ by considering a completeness parameter $\epsilon_c$ as detailed in \cref{comp}. The remaining parameters $v,\alpha',\alpha",\gamma,c_{\perp}$ are optimized in order to maximize the key length $l$.

We note that it is helpful for this optimization to take special care when building the function $g_v$. Details are given in \cref{sec:fi_and_gi}. In \cref{sec:internal_opt,sec:kalpha}, we provide additional insight regarding the practical evaluation of internal optimization in \cref{eq:finitesize} and of the $K_{\alpha'}(f_v)$ term. Finally, we discuss our choice of error correction cost $EC$ in \cref{sec:EC}.

\subsubsection{Optimizing the signaling components}
\label{sec:sigComp}

The main contribution of the signaling components $\vec{\mathrm{f}}$ to the key length $l$ lies in the variance term $\text{Var}_p(\bar f_v)$. In order to limit this contribution, we choose them such that
\begin{equation}
\max_k \bar I(\mathcal{P}_{\delta_k}) - \min_k \bar I(\mathcal{P}_{\delta_k})
\end{equation}
is minimal. Here, $k\in \{(a,b,x,y) \text{ with } a,b\in\{0,1\} \text{ and } x,y\in\{1,2\}\}$ and
\begin{equation}
\mathcal{P}_{\delta_k}(a,b|x,y)=\begin{cases}1 & k=(a,b,x,y) \\ 0 & \text{otherwise}\end{cases}
\end{equation}
is the delta distribution.

\subsubsection{Threshold on the $I$-score \texorpdfstring{$I_{thr}$}{} and \texorpdfstring{$\epsilon$}{}-completeness of the protocol}
\label{comp}

In an experiment using the security proof developed here, we have to put a threshold on the observed $I$-score for the protocol to abort. This requirement is here to ensure secrecy of the protocol: if the observed $I$-score is not high enough, we conclude that the devices are not working as expected. In this case the protocol aborts. We thus have to define a threshold on this $I$-score: $I_{thr}$. When the protocol is carried out and no problem occurs, we expect to have a probability of success close to 1. We decide to put this probability at $1-\epsilon_c = 95\%$. To do so, we apply the Central limit theorem to estimate the value of $I_{thr}$.

In practice, the observation of the statistics is done on the 16 possible events $\{a,b,x,y\}$ of the inputs and outputs of Alice and Bob. Out of the empirical distribution of the observed events, we compute the observed $\bar I$ value, using the vectors $\lambda_k$ and $\vec{\mathrm{f}}$ characterizing the complete $I$-score. We expect that the model fits the experimental behavior and suppose it produces events according to the distribution $\mathcal{P}$. The associated mean value of $\bar I$ is $\mathbb{E}[I] = I(\mathcal{P})$, and its variance is $\sigma(\bar I)$. Using the Central Limit Theorem, we can then build the following random variable

\begin{equation}
    Z_n = \frac{\frac{1}{n} \sum_{i=0}^n I_i - \mathbb{E}[I]}{\sigma(\bar I)/\sqrt{n}} \sim \mathcal{N}(0,1).
\end{equation}
We have $\mathbb{P}(Z_n \geq 1.65) = 0.05$, so if we replace the value of $Z_n$ we get

\begin{equation}
    \mathbb{P}(\frac{1}{n} \sum_i^n I_i \geq \mathbb{E}[I] + \frac{1.65\sigma(\bar I)}{\sqrt{n}}) = 0.05.
\end{equation}
So, using this approximation provided by the central limit theorem, for a $95\%$ success rate of the protocol, we can put $I_{thr} = \mathbb{E}[I] + \frac{1.65\sigma(\bar I)}{\sqrt{n}}$.

\subsubsection{Obtaining the functions \texorpdfstring{$f_v$ and $g_v$}{}}
\label{sec:fi_and_gi}

In order to evaluate the function $g_v$ defined in \cref{eq:mintradeoff_testround}, one needs to compute the derivative of $H(A_1|E)(I)$. We do this numerically by computing $H(A_1|E)(I)$ on an array of points $I$ and constructing an interpolation of the resulting points using Schumaker's spline \cite{Schumaker}. This interpolation has the property of keeping the convexity of the interpolated points \cite{Schumakerarticle}, which ensures numerical stability. Concretely, we use this interpolation to approximation numerically the derivative of $H(A_1|E)$, which is required to obtain $g_v$ and $f_v$.

Additionally, we smoothen out the junction of this function at the boundary of the local set to ease the internal optimisation. Indeed, the last point we compute lies on, or is close to, the local bound, and we know that the function will remain constant beyond that point. Therefore, we smooth the function so that it decreases at most by $\epsilon$ afterward, see Fig.\ref{fig:curve_engineering}. This corresponds to a conservative estimate of the entropy.

\begin{figure}
  \centering
  \includesvg[width=0.3\linewidth]{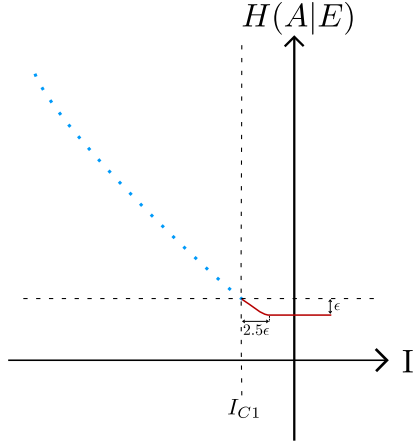}
  \caption{Sketch of the curve smoothing close to the classical bound. The blue dots correspond to the numerical estimation of $H(A|E)$. In the local set, for $\epsilon \geq 0$, we down-shift the value by $\epsilon$, and on the segment $[I_{C_1},I_{C_1}+2.5\epsilon]$ we apply a continuous junction with an polynomial of order three with the last point estimated and the value down-shifted in the local set.}
  \label{fig:curve_engineering}
\end{figure}

To smooth the curve, we define a polynomial of order three to match the interpolated curve to the constant function in the local set, such that the junction is continuous and has the same derivatives (an interpolation of order less than three cannot, in general, ensure smoothness with respect to both the function and its derivative). The function $f_v$ is then obtained using \cref{eq:mintradeoff_full}. We emphasize that this does not affect the final key rate, but only eases the optimization procedure.

Note that we only need to consider the function $H(A|E)(I)$ on the interval $I\in [I_{Q_1},I_{C_1}]$ when bounding the key length in \cref{eq:finitesize}. To see this, notice that $I$ only intervenes in \cref{eq:finitesize} through the first term ($g_v$) and the second term (the internal minimization). Taking an initial point $\mathcal{P}_0$, over which we define the $I$-score and a $I_{thr}$ close to 0 (we have $I(\mathcal{P}_0) = 0$), it means that for any point $\mathcal{P}$, $g_v(I_{thr})$ such that $I(\mathcal{P}) \geq I_{C_1}$ cannot give a positive key length. So we need to choose $v \in [I_{Q_1},I_{C_1}]$. Additionally, in the minimization term, the minimizer $\mathcal{P} \in \mathcal{Q}^{\gamma}$ has to give an $I$-score in this interval, because otherwise the term $\delta(f_v,p)$ would be too large. It means that we can restrain $g_v$ and $f_v$ on the interval $[I_{Q_1},I_{C_1}]$.

\subsubsection{Handling of the internal optimization term}
\label{sec:internal_opt}
The new key length formula in \cref{eq:finitesize} has the following internal optimization term
\begin{equation}
    \inf_{q\in \mathcal{Q}^{\gamma}} ( \Delta (f_v,p(\gamma,q)) - (\alpha-1) V(\bar f_v,p(\gamma,q)) ),
    \label{eq:convexproblem}
\end{equation}
with $p(\gamma,q) = (\gamma q_1, \dots , \gamma q_{16},1-\gamma)$. This optimization is critical, because the length of the key is valid only if the minimum is found. Thankfully, this problem is a convex, as shown in Ref.~\cite{Liu2021}, and so any local minimum is a global minimum.

In order to perform this optimization while staying on the safe side, we relaxing the condition $\mathcal{P} \in Q$, which is difficult to enforce exactly, to admit any distribution $\mathcal{P}$ satisfying positivity, a normalization, non-signaling constraints and $I(q) > I_Q$ with $I_Q$ one of the two quantum bounds. We then optimize over all such behaviors $\mathcal{P}$, we use the algorithm Ipopt \cite{IPopt} with the interface JuMP \cite{Lubin2023}.

\subsubsection{Computation of \texorpdfstring{$\max(f)$}{} and \texorpdfstring{$\min(f)$}{} for the term \texorpdfstring{$K_{\alpha'}(f)$}{}}
\label{sec:kalpha}

The following quantities
\begin{equation}
    \begin{split}
    \label{eq:minmax_terms}
         \max_{\mathcal{P} \in \mathcal{D}_{X_\gamma}}(f) &=  \max\{ \frac{1}{\gamma}  \max_{\mathcal{P} \in \mathcal{D}_{X_\gamma}}(g) + (1-\frac{1}{\gamma})c_{\perp},c_{\perp}  \}\\
         \min_{\mathcal{P} \in \mathcal{Q}^{\gamma}}(f) &= \min_{\mathcal{P} \in \mathcal{Q}^{test}}(g)
    \end{split}
\end{equation}
appear in the formula of the key length given in \cref{eq:finitesize} via the term $K_{\alpha'}(f)$. The first term in \cref{eq:minmax_terms} corresponds to the maximum of $f$ over any distribution that can lie outside of the quantum set. We obtain this quantity by finding $I_{max}$ over all possible distributions, which is done using the algorithm Ipopt \cite{IPopt} with the interface JuMP \cite{Lubin2023}. The second term is the minimum of $f$ on the quantum set; it corresponds to one of the two quantum bounds because the min trade-off function is affine. We show in Fig.\ref{fig:minfmaxf} the associated value of $I$ for these extremal distributions.

\begin{figure}
  \centering
  \includesvg[width=0.5\linewidth]{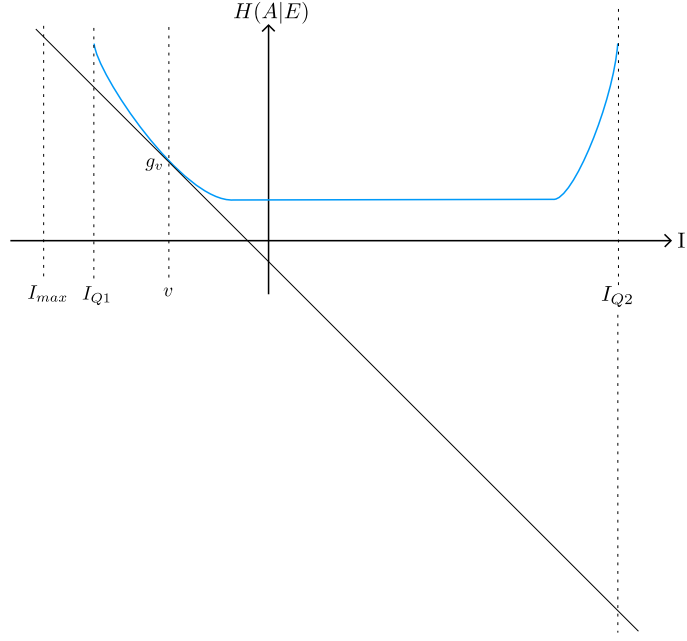}
  \caption{Sketch to show the quantities that have to be computed of $g_v$ for $K_{\alpha'}(f)$ (the tangent of $H(A|E)(I)$ at point $v$), the minimum is taken at $I_{Q2}$ and the maximum at $I_{max}$}
  \label{fig:minfmaxf}
\end{figure}

\subsubsection{Error correction term (EC)}
\label{sec:EC}
In order to estimate the number of bits used for the error correction, we follow the cost estimation given in Ref. \cite{nadlinger_device-independent_2022} for a spatially coupled LDPC code:
\begin{equation}
EC=n((1-\gamma)H(A|B)_{keygen} + \gamma H(A|B)_{test}) + 50 \sqrt{n}.
\end{equation}
Here,
\begin{equation}
    \begin{split}
        H(A|B)_{keygen} &= - \sum_{a,b \in \{0,1\}} \mathcal{P}(a,b|1,0) \log(\frac{\mathcal{P}(a,b|1,0)}{\mathcal{P}_A(a|1)}) \\
        H(A|B)_{test} &= -\frac{1}{4} \sum_{x\in \{1,2\} ,y \in \{1,2\}} \sum_{a,b \in \{0,1\}} \mathcal{P}(a,b|x,y) \log(\frac{\mathcal{P}(a,b|x,y)}{\mathcal{P}_A(a|x)}),
    \end{split}
\end{equation}
where $H(A|B)_{keygen}$ is the conditional entropy during key generation rounds and $H(A|B)_{test}$ during test rounds. Note that in  \cite{nadlinger_device-independent_2022} the entropy terms were directly expressed in terms of the system parameters (CHSH score and quantum bit error rate (QBER)).

\end{document}